\documentclass[english]{article}%
\usepackage{amssymb}
\usepackage[T1]{fontenc}
\usepackage[latin9]{inputenc}
\usepackage[letterpaper]{geometry}
\usepackage{amsmath}
\usepackage{setspace}
\usepackage{babel}
\usepackage{amsfonts}
\usepackage{graphicx}
\usepackage{apacite}%
\providecommand{\U}[1]{\protect\rule{.1in}{.1in}}
\makeatother
\newtheorem{theorem}{Theorem}

\newtheorem{definition}{Definition}
\newtheorem{example}{Example}

\newtheorem{proposition}{Proposition}
\newtheorem{remark}{Remark}

  \makeatother
\newenvironment{proof}[1][Proof]{\textbf{#1.} }{\ \rule{0.5em}{0.5em}}
\begin{document}

\author{Jos\'e Manuel Corcuera\thanks{Universitat de Barcelona, Gran Via de les Corts
Catalanes, 585, E-08007 Barcelona, Spain. \texttt{E-mail: } jmcorcuera@ub.edu.
The work of J. M. Corcuera is supported by the Spanish grant PID
118339GB-I00.} \hspace{1mm} and \hspace{1mm} Giulia Di Nunno\thanks{University
of Oslo, Department of Mathematics, P.O. Box 1053 Blindern NO-0316 Oslo, and
NHH, Department of Business and Management Science, Helleveien 30, 5045
Bergen, Norway. \texttt{E-mail: }giulian@math.uio.no. This work has received
support from the Research Council of Norway via the project \textit{STORM:
Stochastics for Time-Space Risk Models} (nr. 274410).} }
\title{Path-dependent Kyle equilibrium model.}
\date{July 26, 2024}
\maketitle

\begin{abstract}
\noindent We consider an auction type equilibrium model with an insider in
line with the one originally introduced by Kyle in 1985 and then extended to
the continuous time setting by Back in 1992. The novelty introduced with this
paper is that we deal with a general price functional depending on the
\textit{whole} past of the aggregate demand, i.e. we work with
path-dependency. By using the functional It\^o calculus, we provide necessary
and sufficient conditions for the existence of an equilibrium. Furthermore, we
consider both the cases of a risk-neutral and a risk-averse insider.

\vspace{1mm} \noindent\textbf{Key words}: Kyle model, market microstructure,
equilibrium, insider trading, stochastic control, semimartingales, functional
It\^o calculus.

\noindent\textbf{JEL-Classification}{\ C61$\cdot$ D43$\cdot$ D44$\cdot$
D53$\cdot$ G11$\cdot$ G12$\cdot$ G14}

\noindent\textbf{MS-Classification 2020}: 60G35, 62M20, 91B50, 93E03

\end{abstract}

\vspace{-7mm}

\section{Introduction}

It is well known that insider information and informational asymmetries are
everywhere in the real economy. In his pioneering work, \citeA{kyle85}
constructed a model in a discrete time setting with market makers, uninformed
traders and one insider, who knows the fundamental value of an asset at a
certain fixed released time. Also, the model included a price functional
relating market prices and the total demand. \citeA{back92}, extends Kyle's
model to the continuous time case. Since these worked appeared, several
generalisations and extensions have been produced. To mention some,
\citeA{bacped98}, who consider a \emph{dynamic} fundamental price and Gaussian
noises with time varying volatility; \citeA{cho03}\ who considers pricing
functions depending on the path of the demand process and also studies the
case when the informed trader is risk-averse; \citeA{lasserre04}, who
considers a multivariate setting; \citeA{bacbar04}, where the market depth
(i.e. the marginal effect on price of the volume traded) depends on the market
price of the stock; \citeA{aabjok12a,aabjok12b}, who put emphasis on filtering
techniques to solve the equilibrium problem; \citeA{camcet07}, who consider a
defaultable bond instead of a stock as in the Kyle-Back model and also
consider the knowledge of the default time as the insider's privileged
information; \citeA{Danilova10}, who deals with non-regular pricing rules;
\citeA{calsta10} who take a random release time into account;
\citeA{caluda13}, who consider again a defaultable bond, but this time they
consider the privileged information to be represented by some dynamic signal
related with the default time; and \citeA{colfos16} where the market depth
depends on the (random) volatility of the noise in the market. In
\citeA{cordin18}, the authors propose a general framework to include all the
particular extensions mentioned above and study the general characteristics of
the equilibria. Recently \citeA{codifa19} have also considered the same
general situation, but with a random price pressure and a random release time
of information.

In this paper we propose a step even further and we consider a price
functional that depends on the whole path of the aggregate demand. We study
the properties of the equilibrium and sufficient and necessary conditions for
the existence of equilibria. Also we study both the case of a risk-neutral and
risk-averse insider. With this work we extend substantially the present
frontiers of the literature on this theme. We also note that the analysis of
price functional of these type will involve the recently introduced functional
calculus, see e.g. \citeA{confou13}, and this work represents a good venture
to see these new mathematical techniques applied in economics and finance.

The paper is structured as follows. In the next section we describe the model
and we define the equilibrium for the admissible strategies. Section 3
presents some needed background material from functional It\^o calculus. In
the Section 4, we also suggest some general results that allow to reduce the
set of admissible strategies on which the insider can find its optimal
performance and hence describe the equilibrium. Also, we give necessary and
sufficient conditions to obtain an equilibrium under very general classes of
pricing rules. In this section, we consider the two cases of a risk-neutral
and a risk adverse insider. Section 5 is devoted to the study of necessary
conditions for an equilibrium, without fixing a priori, up to smoothness
conditions, the set of pricing rules. We observe that these latter results
motivate and justify those restrictions imposed on the classes of pricing
rules considered in the study of Section 4. The last two sections are
dedicated to examples of classes of pricing rules and examples of equilibrium
models, correspondingly.

\section{The model and equilibrium}

We consider a market with two assets, a stock and a bank account with interest
rate $r$ equal to zero for the sake of simplicity. The trading is continuous
in time over the period $[0,\infty)$ and it is order driven. There is a
(possibly random) release time of information $\ \tau<\infty$ a.s., when the
fundamental value of the stock is revealed. The \textit{fundamental value
process} represents the actual value of the asset, which would be the same as
the \textit{market price} of the asset only if \emph{all} the information was
public. We could say, with \citeA{malkiel11}, that the fundamental value is
the intrinsic value of a stock, via an analysis of the balance sheet, the
expected future dividends, and the growth prospects of a company. The
fundamental value process is denoted by $V$.

We shall denote the market price of the stock at time $t$ by $P_{t}.$ This
represents the market evaluation of the asset. Just after the revelation time
$\tau,$ the price of the stock coincides with the fundamental value. Then we
consider $P_{t}$ \ defined only on $t\leq\tau$. Obviously, it is possible that
$P_{t}\neq V_{t}$ for $t\leq\tau$.

We assume that all the random variables and processes mentioned are defined in
the same {complete} filtered probability space $(\Omega,\mathcal{F}%
,\mathbb{H},\mathbb{P})$ where the filtration $\mathbb{H}$ and any other
filtration considered in this present work are {complete and right-continuous
}by taking, when necessary, the usual augmentation.

There are three kinds of traders. A number of \textit{liquidity traders}, who
trade for liquidity or hedging reasons, an \textit{informed trader} or
\textit{insider}, who has privileged information about the firm and can deduce
its fundamental value, and the \textit{market makers}, who set the market
price and clear the market.

\subsection{The agents and the equilibrium\label{Agents}}

At time $t$, the insider's information is the full information $\mathcal{H}%
_{t}$ and her flow of information is represented by the filtration
$\mathbb{H}=\left(  \mathcal{H}_{t}\right)  _{t\geq0}$. Since this is also the
filtration with respect to which all the processes considered in the present
work are adapted, we shall omit to write it in the notation. A random release
time of information $\tau$ is considered from insider's perspective to be of
one of these types:

\begin{itemize}
\item $\tau$ it is bounded and predictable,

\item $\tau$ it is not a predictable stopping time, but it is independent of
the observable variables.
\end{itemize}

We assume that the fundamental value $V$ is a continuous martingale such that
$\sigma_{V}^{2}(t):=\frac{\mathrm{d}[V,V]_{t}}{\mathrm{d}t}$ is well defined.

Hereafter we describe in detail the three types of agents involved in this
market model, namely their role, their demand process, and their information.
Let $Z$ be the \emph{aggregate} demand process of the liquidity traders. We
recall that there are a large number of traders motivated by liquidity or
hedging reasons. They are perceived by the insider as constituting noise in
the market, thus also called \emph{noise} traders. It is assumed that $Z$ is a
\emph{continuous} martingale, starting at zero, independent of $V$, and such
that $\sigma_{Z}^{2}(t):=\frac{\mathrm{d}[Z,Z]_{t}}{\mathrm{d}t}$ is well
defined. As it is shown in \citeA{corcuera2010kyle}, if $Z$ had jumps, an
equilibrium would not be possible.

\begin{remark}
In this equilibrium model, the time $\tau$ and the processes $V$ and $Z$ are
exogenously given.
\end{remark}

Market makers clear the market giving the market prices. They rely on the
information given by the total aggregate demand $Y$, which they observe, and
the release time $\tau$, that is a stopping time for them. Hence, their
information flow is: $\mathbb{F}=\left(  \mathcal{F}_{t}\right)  _{t\geq0}$,
where $\mathcal{F}_{t}=\bar{\sigma}(Y_{s},\tau\wedge s,0\leq s\leq t)$. \ Here
$\bar{\sigma}$ denotes the $\sigma$-field corresponding to the usual
augmentation of the natural filtration.

The total aggregate demand is defined as $Y:=X+Z$, where $X$ denotes the
insider demand process, which is naturally assumed to be a predictable process
and also a càdlàg semimartingale:
\begin{align*}
\text{\textbf{(A1)}}\qquad &  X_{t}=M_{t}+A_{t}+\int_{0}^{t}\theta
_{s}\mathrm{d}s,\quad t\geq0,\\
&  \text{where }M\text{ is a continuous martingale with }M_{0}=0\text{, }\\
&  A\text{ a bounded variation}\text{ predictable process with }A_{t}%
=\sum_{0<s\leq t}\left(  X_{s}-X_{s-}\right)  \text{ and }A_{0}=0,\\
&  \theta\text{ is a càdlàg adapted process.}%
\end{align*}
Strategies $X$ satisfying \textbf{(A.1)} are called \emph{admissible}. Market
makers provide liquidity and fix the market prices $P_{t}$, for all $t$, based
on the total demand $Y$, resulting in the functionals:
\[
P_{t}=P_{t}(Y_{s},0\leq s\leq t),\quad0\leq t\leq\tau.
\]
It is natural to assume that prices are strictly increasing with the total
demand $Y$. We shall precise this condition in the next section.

From the economic point of view, due to the competition among market makers,
the market prices are \emph{competitive}, in the sense that
\begin{equation}
\label{competitive}P_{t}=\mathbb{E}(V_{t}|\mathcal{F}_{t}),\quad0\leq
t\leq\tau.
\end{equation}
Therefore $\left(  P_{t}\right)  _{0\leq t\leq\tau}$ is an $\mathbb{F}$-martingale.

\begin{definition}
The couple $(P,X)$ is an \emph{equilibrium} if market prices admit a pricing
rule (i.e. a functional of $Y$), that we shall name \emph{equilibrium pricing
rule},
\[
P_{t}=P_{t}(Y_{s},0\leq s\leq t),\quad0\leq t\leq\tau.
\]
such that, at the same time, the market prices $P$ are competitive given $X$,
i.e.
\[
P_{t}=\mathbb{E}(V_{t}|\mathcal{F}_{t}),\quad0\leq t\leq\tau,
\]
and the strategy $X$ is optimal for the insider given the prices $P$.
\end{definition}

Now we have to make precise what an optimal strategy for the insider is. The
informed trader aims at maximizing the expected final utility of her wealth.
Let $W$ be the wealth process corresponding to the insider's portfolio $X$. To
obtain the formula for the insider's wealth assume that trades occur at times
$0\leq t_{1}\leq t_{2}\leq...\leq t_{N}=\tau.$ If at time $t_{i-1}$ there is
an order to buy $X_{t_{i}}-X_{t_{i-1}}$ shares, its \emph{cost} will be
$P_{t_{i}}\times(X_{t_{i}}-X_{t_{i-1}})$, so there is a change in the
insider's bank account given by
\[
-P_{t_{i}}\times(X_{t_{i}}-X_{t_{i-1}})=-P_{t_{i-1}}\times(X_{t_{i}%
}-X_{t_{i-1}})-\left(  P_{t_{i}}-P_{t_{i-1}}\right)  \times(X_{t_{i}%
}-X_{t_{i-1}}),
\]
where the second term in the right-hand side accounts for the impact of the
demand on the current price. Due to the fact that the price of the asset
equals its fundamental value at the release time $\tau$, there is, in
addition, the extra income $X_{\tau}V_{\tau}$. Then the total wealth at $\tau$
is given by
\[
W_{\tau}=-\sum_{i=1}^{N}P_{t_{i-1}}\times(X_{t_{i}}-X_{t_{i-1}})-\sum
_{i=1}^{N}\left(  P_{t_{i}}-P_{t_{i-1}}\right)  \times(X_{t_{i}}-X_{t_{i-1}%
})+X_{\tau}V_{\tau},
\]
so taking the limit with the time between trades going to zero, we have
\begin{equation}
W_{\tau}=-\int_{0}^{\tau}P_{t-}\mathrm{d}X_{t}-[P,X]_{\tau}+X_{\tau}V_{\tau
}\nonumber
\end{equation}
where (here and throughout the whole article) $P_{t-}:=\lim_{s\uparrow t}%
P_{s}$ a.s.

Then the informed trader aims at maximizing
\begin{equation}
\mathbb{E}(\left.  U\left(  W_{\tau}\right)  \right\vert \mathcal{H}%
_{0})=\mathbb{E}\left(  \left.  U\left(  -\int_{0}^{\tau}P_{t-}\mathrm{d}%
X_{t}-[P,X]_{\tau}+X_{\tau}V_{\tau}\right)  \right\vert \mathcal{H}%
_{0}\right)  \label{Wnew}%
\end{equation}
for a given utility function $U$, that is, a strictly increasing and concave
function satisfying the Inada conditions. The case when $U$ is the identity
function corresponds to the so called \emph{risk-neutral} case. The insider's
strategy $X$ of type \textbf{(A.1)} providing the maximum is called
\emph{optimal}.

\section{Regularity of the functionals. The functional It\^o formula}

Trading is developed in the context of \emph{imperfect competition}, in the
sense that prices are affected by the demand, that is $P_{t}=P_{t}(Y_{s},0\leq
s\leq t).$ Here and in the sequel, we shall write $Y_{\cdot t}$ to indicate
the path of the process $Y$ from zero to $t$:%
\[
Y_{\cdot t}(s):=Y_{s},\quad0\leq s\leq t.
\]
Notice that we also can look at $Y_{\cdot t}$ as the process $Y$ stopped at
$t$, in such a way that
\[
Y_{\cdot t}(s):=\left\{
\begin{array}
[c]{cc}%
Y_{s} & \text{for }0\leq s\leq t\\
Y_{t} & \text{for }t\leq s\leq\tau
\end{array}
\right.
\]
Therefore we can write, alternatively, $P_{t}=P_{t}(Y_{\cdot t})$ and to
consider $P_{t}$ as a functional of the process $Y$ stopped at $t$. We shall
also add some regularity on the functionals we are going to consider when
needed. We shall also consider the following perturbation of a process $Y$.
For $h\in\mathbb{R}$, we define
\[
Y_{\cdot t}^{h}(s):=\left\{
\begin{array}
[c]{cc}%
Y_{s} & \text{for }0\leq s<t\\
Y_{t}+h & \text{for }t\leq s\leq\tau
\end{array}
\right.  .
\]

\bigskip We are going to consider \emph{non-anticipative} functionals
$F:\Lambda\subseteq\lbrack0,T]\times D\left(  [0,T],\mathbb{R}\right)
\rightarrow\mathbb{R}$ where%
\[
\Lambda=\left\{  \left(  t,y_{.t}\right)  ,t\in\lbrack0,T],y_{\cdot
t}(s)=y_{s\wedge t},s\in\lbrack0,T]\right\}
\]
in such a way that if $F(t,\cdot)\left(  \equiv F_{t}\left(  \cdot\right)
\right)  $ is a measurable map (with the canonical filtration in $D\left(
[0,T],\mathbb{R}\right)  $) then $F_{t}\left(  Y_{\cdot t}\right)  $ is an
adapted process.

In this context we can define the \emph{horizontal} derivative
(\citeA{dupire2009functional})
\[
\mathcal{D}_{t}F_{t}\left(  y_{\cdot t}\right)  =\lim_{\Delta t\downarrow
0}\frac{F_{t+\Delta t}(y_{\cdot t})-F_{t}(y_{\cdot t})}{\Delta t}.
\]

\begin{example}
Consider $G,F$ smooth \textbf{functions} and $h,f$ integrable
\textbf{functions} and $Y$ \ an Itô process
\begin{align*}
\mathcal{D}_{t}G(Y_{t})  &  =0,\mathcal{D}_{t}F(t,Y_{t})=\partial_{t}%
F(t,Y_{t})\\
\mathcal{D}_{t}\int_{0}^{t}h(Y_{s})\mathrm{d}s  &  =h(Y_{t}),\text{
}\mathcal{D}_{t}\left(  \int_{0}^{t}f(s,Y_{s})\mathrm{d}Y_{s}\right)  =?.
\end{align*}
$\mathcal{D}_{t}\left(  \int_{0}^{t}f(s,Y_{s})\mathrm{d}Y_{s}\right)  =?$ Let
$F(t,y)\in C^{1,2}$, with $f(t,y)=\partial_{y}F(t,y),$ by the classical Itô
formula,%
\begin{align*}
\int_{0}^{t}f(s,Y_{s})\mathrm{d}Y_{s}  &  =F(t,Y_{t})-F(0,Y_{0})-\int_{0}%
^{t}\partial_{s}F(s,Y_{s})\mathrm{d}s\\
&  -\frac{1}{2}\int_{0}^{t}\partial_{y}f(s,Y_{s})\mathrm{d}[Y,Y]_{s}.
\end{align*}
Assume that $\mathrm{d}[Y,Y]_{t}=A_{t}\mathrm{d}t,$then we are tempted to
write
\[
\mathcal{D}_{t}\left(  \int_{0}^{t}f(s,Y_{s})\mathrm{d}X_{s}\right)
=-\frac{1}{2}\partial_{y}f(t,Y_{t})A_{t}.
\]
However
\[
\int_{0}^{t+\Delta t}f(s,Y_{s\wedge t})\mathrm{d}Y_{s\wedge t}=\int_{0}%
^{t}f(s,Y_{s})\mathrm{d}Y_{s},
\]
so\ $\mathcal{D}_{t}\left(  \int_{0}^{t}f(s,Y_{s})\mathrm{d}Y_{s}\right)  =0$
if we apply the previous definition!
\end{example}

We can also define the \emph{vertical} derivative
\[
\mathcal{\nabla}_{x}F_{t}\left(  y_{\cdot t}\right)  :=\lim_{h\rightarrow
0}\frac{F_{t}\left(  y_{\cdot t}^{h}\right)  -F_{t}\left(  y_{\cdot t}\right)
}{h},
\]

\begin{example}%
\begin{align*}
\mathcal{\nabla}_{x}G(Y_{t})  &  =\partial_{y}G(Y_{t}),\mathcal{\nabla}%
_{x}F(t,Y_{t})=\partial_{y}F(t,Y_{t})\\
\mathcal{\nabla}_{y}\int_{0}^{t}h(Y_{s})\mathrm{d}s  &  =0,\text{
}\mathcal{\nabla}_{y}\left(  \int_{0}^{t}f(s,Y_{s})\mathrm{d}Y_{s}\right)  =?.
\end{align*}
$\mathcal{\nabla}_{y}\left(  \int_{0}^{t}f(s,Y_{s})\mathrm{d}Y_{s}\right)  =?$
As above,
\begin{align*}
\int_{0}^{t}f(s,Y_{s})\mathrm{d}Y_{s}  &  =F(t,Y_{t})-F(0,Y_{0})-\int_{0}%
^{t}\partial_{s}F(s,Y_{s})\mathrm{d}s\\
&  -\frac{1}{2}\int_{0}^{t}\partial_{y}f(s,Y_{s})A_{s}\mathrm{d}s,
\end{align*}
then
\[
\mathcal{\nabla}_{y}\left(  \int_{0}^{t}f(s,Y_{s})\mathrm{d}Y_{s}\right)
=\partial_{y}F(t,Y_{t})=f(t,Y_{t}).
\]

\end{example}

These derivatives satisfy the usual properties: linearity, product rule and
chain rule. However, in general, they do not commute. Set%
\[
\mathcal{L}_{t}:=\mathcal{D}_{t}\bigtriangledown_{y}-\bigtriangledown
_{y}\mathcal{D}_{t},
\]

\begin{example}%
\begin{align*}
\mathcal{L}_{t}\left(  \int_{0}^{t}f(s,Y_{s})\mathrm{d}s\right)   &
=\mathcal{D}_{t}\bigtriangledown_{y}\left(  \int_{0}^{t}f(s,Y_{s}%
)\mathrm{d}s\right) \\
&  -\bigtriangledown_{y}\mathcal{D}_{t}\left(  \int_{0}^{t}f(s,Y_{s}%
)\mathrm{d}s\right) \\
&  =-\bigtriangledown_{y}f(t,Y_{t})
\end{align*}%
\begin{align*}
\mathcal{L}_{t}\left(  \int_{0}^{t}f(s,Y_{s})\mathrm{d}Y_{s}\right)   &
=\left(  \mathcal{D}_{t}\bigtriangledown_{y}-\bigtriangledown_{y}%
\mathcal{D}_{t}\right)  \left(  \int_{0}^{t}f(s,Y_{s})\mathrm{d}Y_{s}\right)
\\
&  =\mathcal{D}_{t}\bigtriangledown_{y}\left(  \int_{0}^{t}f(s,Y_{s}%
)\mathrm{d}Y_{s}\right) \\
&  -\bigtriangledown_{y}\mathcal{D}_{t}\left(  \int_{0}^{t}f(s,Y_{s}%
)\mathrm{d}Y_{s}\right) \\
&  =\mathcal{D}_{t}f(t,Y_{t})+\frac{1}{2}\bigtriangledown_{y}^{2}%
f(t,Y_{t})A_{t}\neq0
\end{align*}
if $\left(  f(t,Y_{t})\right)  _{t\geq0}$ is not a local martingale.\bigskip
\end{example}

Given two stopped processes $Y_{.t},Z_{.t^{\prime}}$ we consider the distance
defined by
\[
d_{\infty}(Y_{\cdot t},Z_{\cdot t^{\prime}})=\left\Vert Y_{\cdot t}-Z_{\cdot
t^{\prime}}\right\Vert _{\infty}+|t-t^{\prime}|.
\]
where $\left\Vert \cdot\right\Vert _{\infty}$ is the sup-norm$.$

\begin{definition}
A non-anticipative functional $P$ is said to be \emph{left-continuous} at $t$
if for all $\varepsilon>0$ there exists $\eta>0$ such that for all $0\leq
t^{\prime}\leq t\leq\tau$%
\[
d_{\infty}(Y_{\cdot t},Z_{\cdot t^{\prime}})<\eta\Longrightarrow\left\vert
P_{t}(Y_{\cdot t})-P_{t^{\prime}}(Z_{\cdot t^{\prime}})\right\vert
<\varepsilon
\]

\end{definition}

$P$ is said to be \emph{left-continuous} if it is left-continuous at any
$\left(  t,Y_{\cdot t}\right)  $. \emph{Right-continuity} is defined
analogously. \emph{Continuity} means that left- and right-continuity occur at
the same time. If, in the previous definition, we consider only times
$t^{\prime}=$ $t$ then we say that the functional is said to be
\emph{continuous at fixed times}.

Since the space of c\`adl\`ag functions is not separable under the sup-norm,
we need the following additional regularity, even for the continuous
functionals defined above.

\begin{definition}
A functional $P$ is said to be \emph{boundedness preserving} if for every
constant $K\ $and $t_{0}\leq T$ there exists a constant $C_{K,t_{0}}$ such
that for all $t\leq t_{0}\leq T,$with $\left\Vert Y_{\cdot t}\right\Vert
_{\infty}<K$%
\[
\quad\left\vert P_{t}(Y_{\cdot t})\right\vert <C_{K,t_{0}}.
\]

\end{definition}

\begin{definition}
We say that a left-continuous functional belongs to $\mathbb{C}_{b}^{j,k}$ if
it is $j$-times horizontally differentiable with derivatives continuous at
fixed points and boundedness preserving, and it is $k$-times vertically
differentiable with left-continuous and boundedness preserving derivatives.
\end{definition}

\begin{theorem}
\label{FIT copy(1)}\emph{(Functional Itô's formula (Cont-Fournié))}. If $Y$ is
a càdllàg semimartingale and $P$ $\in\mathbb{C}_{b}^{1,2}$ then%
\begin{align*}
&  P_{t}(Y_{t})\\
&  =P_{0}(Y_{0})+\int_{0}^{t}\mathcal{D}_{s}P_{s}(Y_{\cdot s})\mathrm{d}%
s+\int_{0}^{t}\mathcal{\nabla}_{Y}P_{s}(Y_{\cdot s})\mathrm{d}Y_{s}\\
&  +\frac{1}{2}\int_{0}^{t}\mathcal{\nabla}_{Y}^{2}P_{s}(Y_{\cdot
s})\mathrm{d}[Y^{c},Y^{c}]_{s}\\
&  +\sum_{\substack{0<u<t\\\Delta Y\neq0}}\left(  P_{u}(Y_{\cdot u}%
)-P_{u}(Y_{\cdot u-})-\mathcal{\nabla}_{Y}P_{u-}\Delta Y_{u}\right)  .
\end{align*}
\bigskip
\end{theorem}

A first implication of these results is the following proposition.

\begin{proposition}
Assume that the functional $P$ is left-continuous, boundedness preserving,
belongs to $\mathbb{C}_{b}^{1,1}$ and strictly increasing, that is $\nabla
_{Y}P_{t}>0,$ then admissible strategies $X$ with a continuous martingale part
or jumps are suboptimal in the class of all admissible strategies.
\end{proposition}

\begin{proof}
Let $F(t,Y_{\cdot t})$ be a smooth functional, $F\in\mathbb{C}_{b}^{1,2}$,
with $\mathcal{\nabla}_{Y}F_{t}=P_{t}$. Then, assume that $X$ has a continuous
martingale part and jumps. Take, by simplicity $\tau\equiv T$\ and
$V_{t}\equiv V$,%
\begin{align*}
F\left(  T,Y_{\cdot T}\right)  -F\left(  0,Y_{0}\right)   &  =\int_{0}%
^{T}P_{t-}\mathrm{d}Y_{t}+\int_{0}^{T}\partial_{t}F\mathrm{d}t+\frac{1}{2}%
\int_{0}^{T}\nabla_{Y}P_{t}\mathrm{d}[Y^{c},Y^{c}]_{t}\\
&  +\sum_{t:\Delta X\neq0}F(t,Y_{\cdot t})-F(t,Y_{\cdot t-})-\nabla_{Y}%
F_{t}(t,Y_{\cdot t-})\Delta Y_{t},
\end{align*}
now, since
\[
\mathrm{d}[Y^{c},Y^{c}]_{t}=\mathrm{d}[X^{c},X^{c}]_{t}+2\mathrm{d}%
[X^{c},Z]_{t}+\mathrm{d}[Z,Z]_{t}%
\]
and
\[
\lbrack P,X]=\int_{0}^{T}\nabla_{Y}P_{t}\mathrm{d}[X^{c},Z]_{t}+\int_{0}%
^{T}\nabla_{Y}P_{t}\mathrm{d}[X^{c},X^{c}]_{t}+\sum_{t:\Delta X\neq0}\Delta
P_{t}\Delta X_{t}%
\]
we have
\begin{align*}
W_{T}  &  =-\int_{0}^{T}P_{t-}\mathrm{d}X_{t}-[P,X]+X_{T}V\\
&  =-F\left(  T,Y_{\cdot T}\right)  +F\left(  0,Y_{0}\right)  +\int_{0}%
^{T}\partial_{t}F\mathrm{d}t\\
&  +\frac{1}{2}\int_{0}^{T}\nabla_{Y}P_{t}\mathrm{d}[X^{c},X^{c}]_{t}+\int%
_{0}^{T}\nabla_{Y}P_{t}\mathrm{d}[X^{c},Z]_{t}+\frac{1}{2}\int_{0}^{T}%
\nabla_{Y}P_{t}\mathrm{d}[Z,Z]_{t}\\
&  -\int_{0}^{T}\nabla_{Y}P_{t}\mathrm{d}[X^{c},Z]_{t}-\int_{0}^{T}\nabla
_{Y}P_{t}\mathrm{d}[X^{c},X^{c}]_{t}-\sum_{t:\Delta X\neq0}\Delta P_{t}\Delta
X_{t}+X_{T}V\\
&  \left.  +\int_{0}^{T}P_{t}\mathrm{d}Z_{t}+\sum_{t:\Delta X\neq0}\left(
F(t,Y_{\cdot t})-F(t,Y_{\cdot t-})-P_{t-}\Delta Y_{t}\right)  \right.
\end{align*}
Therefore%
\begin{align*}
&  W_{T}\\
&  =-F\left(  T,Y_{T\cdot}\right)  +F\left(  0,Y_{0}\right)  +\int_{0}%
^{T}\partial_{t}F\mathrm{d}t+\frac{1}{2}\int_{0}^{T}\nabla_{Y}P_{t}%
\mathrm{d}[Z,Z]_{t}\\
&  +X_{T}V+\int_{0}^{T}P_{t}\mathrm{d}Z_{t}-\frac{1}{2}\int_{0}^{T}\nabla
_{Y}P_{t}\mathrm{d}[X^{c},X^{c}]_{t}\\
&  +\sum_{t:\Delta X\neq0}\left(  F(t,Y_{\cdot t})-F(t,Y_{\cdot t-}%
)-P_{t}\Delta Y_{t}\right)  ,
\end{align*}
and the contribution of the last two terms is always negative because
$\nabla_{Y}^{2}F_{t}=\nabla_{Y}P_{t}>0$ and the other terms can be approximate
as far as we want by replacing $Y=Z+X$ by $\tilde{Y}=Z+\tilde{X}$ where
$\tilde{X}$ is an absolutely continuous approximation to $X.$\bigskip
\end{proof}

\begin{remark}
The previous result is not true if we consider, for instance, price processes
of the form
\[
P_{t}=e^{Y_{t}-\frac{1}{2}[Y]_{t}},t\geq0,
\]
where $Y$ is a continuous semimartingale. The reason is that we can not
approximate $[X]_{t}$ by $[\tilde{X}]_{t}$ $\ $when $\tilde{X}$ is an
absolutely continuous approximation to $X.$ In other words $[X]_{t}$ is not a
continuous functional of $X$ and we need to extend the regularity notion of
functionals in order to work with that kind of functionals.
\citeA{ccetin2021pricing} showed that in fact absolutely continuous strategies
are suboptimal and that the insider can get an unbounded wealth by using
strategies with a continuous martingale part.\bigskip
\end{remark}

\bigskip To include these functionals in the class of regular ones we shall
consider non-anticipative functionals as maps
\[
P:(t,Y_{\cdot t},A_{\cdot t})\rightarrow P(t,Y_{\cdot t},A_{\cdot t}%
)=P_{t}(Y_{\cdot t},A_{\cdot t})
\]
where $A_{t}:=\frac{\mathrm{d}[Y^{c},Y^{c}]_{t}}{\mathrm{d}t},$ and given two
stopped processes $Y_{.t},Z_{.t^{\prime}}$ we consider the distance defined
by
\[
d_{\infty}(\left(  Y_{\cdot t},A_{\cdot t}\right)  ,\left(  Z_{\cdot
t^{\prime}},B_{\cdot t^{\prime}}\right)  )=\left\Vert \left(  Y_{\cdot
t}-Z_{\cdot t^{\prime}},A_{\cdot t}-B_{\cdot t^{\prime}}\right)  \right\Vert
_{\infty}+|t-t^{\prime}|.
\]
where $B_{t}:=\frac{\mathrm{d}[Z^{c},Z^{c}]_{t}}{\mathrm{d}t}$ and $\left\Vert
\cdot\right\Vert _{\infty}$ is the sup-norm$.$ We have analogous definitions
to the above case where functionals depended only on $Y_{\cdot t}$.

\begin{definition}
A nonanticipative functional $P$ is said to be \emph{left-continuous} at if
for all $\varepsilon>0$ there exists $\eta>0$ such that for all $0\leq
t^{\prime}\leq t\leq\tau$%
\[
d_{\infty}(\left(  Y_{\cdot t},A_{\cdot t}\right)  ,\left(  Z_{\cdot
t^{\prime}},B_{\cdot t^{\prime}}\right)  )<\eta\Longrightarrow\left\vert
P_{t}(Y_{\cdot t},A_{\cdot t})-P_{t^{\prime}}(Z_{\cdot t^{\prime}},B_{\cdot
t^{\prime}})\right\vert <\varepsilon
\]

\end{definition}

$P$ is said to be \emph{left-continuous} if it is left-continuous at any
$\left(  t,Y_{\cdot t},A_{.t}\right)  $. \emph{Right-continuity} is defined
analogously. \emph{Continuity} means that left- and right-continuity occur at
the same time. If, in the previous definition, we consider only times
$t^{\prime}=$ $t$ then we say that the functional is said to be
\emph{continuous at fixed times}. It can be seen that continuity at fixed
times implies that the process $\left(  P_{t}(Y_{\cdot t},A_{\cdot t})\right)
_{0\leq t\leq T\text{ }}$ is adapted if $Y$ is adapted.

As above, since the space of c\`adl\`ag functions is not separable under the
sup-norm, we need the following additional regularity, even for the continuous
functionals defined.

\begin{definition}
A functional $P$ is said to be \emph{boundedness preserving} if for every
constants $K$, $R$ $\ $and $t_{0}\leq T$ there exists a constant
$C_{K,R,t_{0}}$ such that for all $t\leq t_{0}\leq T,$with $\left\Vert
Y_{\cdot t}\right\Vert _{\infty}<K$%
\[
\left\Vert A_{\cdot t}\right\Vert _{\infty}<R\quad\Rightarrow\quad\left\vert
P_{t}(Y_{\cdot t},A_{\cdot t})\right\vert <C_{K,R,t_{0}}.
\]

\end{definition}

\begin{definition}
We call \emph{horizontal derivative of the functional $P$ at $\left(
t,Y_{\cdot t},A_{\cdot t}\right)  $}, the limit given by
\[
\mathcal{D}_{t}P_{t}:=\lim_{\Delta t\downarrow0}\frac{P_{t+\Delta t}(Y_{\cdot
t},A_{\cdot t})-P_{t}(Y_{\cdot t},A_{\cdot t})}{\Delta t},
\]
provided it exists.
\end{definition}

\begin{definition}
We call \emph{vertical derivative of the functional $P$ at $\left(  t,Y_{\cdot
t}\right)  $} the limit, provided it exists, given by
\[
\mathcal{\nabla}_{Y}P_{t}:=\lim_{h\rightarrow0}\frac{P_{t}(Y_{\cdot t}%
^{h},A_{\cdot t})-P_{t}(Y_{\cdot t},A_{\cdot t})}{h}.
\]

\end{definition}

\begin{definition}
We say that a left-continuous functional belongs to $\mathbb{C}_{b}^{j,k}$ if
it is $j$-times horizontally differentiable with derivatives continuous at
fixed points and boundedness preserving, and it is $k$-times vertically
differentiable with left-continuous and boundedness preserving derivatives.
\end{definition}

\begin{theorem}
\label{FIT}\emph{(Functional It\^o's formula)}. If $Y$ is a continuous
semimartingale and $P$ $\in\mathbb{C}_{b}^{1,2}$with $P_{t}(Y_{\cdot
t},A_{\cdot t})=P_{t}(Y_{\cdot t},A_{\cdot t-})$, then%
\[
P_{t}(Y_{\cdot t},A_{\cdot t})=P_{0}(Y_{0},A_{0})+\int_{0}^{t}\mathcal{D}%
_{s}P_{s}(Y_{\cdot s},A_{\cdot s})\mathrm{d}s+\int_{0}^{t}\mathcal{\nabla}%
_{Y}P_{s}(Y_{\cdot s},A_{\cdot s})\mathrm{d}Y_{s}+\frac{1}{2}\int_{0}%
^{t}\mathcal{\nabla}_{Y}^{2}P_{s}(Y_{\cdot s},A_{\cdot s})\mathrm{d}%
[Y,Y]_{s},\text{ }\mathbb{P}\text{-a.s. }0\leq t\leq\tau.
\]

\end{theorem}

\begin{proof}
See Theorem 4.1 in \citeA{confou13}
\end{proof}

In the following we shall assume that $P$ is \emph{strictly increasing for all
}$\left\{  Y_{s},0\leq s\leq t\right\}  ,$ that is $\mathcal{\nabla}_{Y}%
P_{t}\left(  Y_{\cdot t}\right)  >0$ for all $\left\{  Y_{s},0\leq s\leq
t\right\}  $ and $t\leq\tau$.

\section{Necessary and sufficient conditions for an equilibrium}

In this section we present necessary and sufficient conditions for the
existence of an equilibrium when the release time $\tau$ and the pricing
functional satisfy some conditions. The nature of these conditions will be
further studied in the next section. In this analysis we shall consider both a
risk-neutral insider and a risk-averse insider with exponential utility function.

In the sequel we shall reduce the set of \emph{admissible insider's
strategies} \textbf{(A1)} to those strategies $X$ satisfying
\[
\mathbf{(\mathbf{A1}^{\prime})}\qquad X_{t}=\int_{0}^{t}\theta_{s}%
\mathrm{d}s,\text{ for all }t\geq0\text{, where }\theta\text{ is a càdlàg
adapted process. }%
\]
Furthermore, the goal of the insider becomes to maximise the performance
\begin{equation}
J(X):=\mathbb{E}\left(  \left.  U(W_{\tau})\right\vert \mathcal{H}_{0}\right)
=\mathbb{E}\left(  \left.  U\left(  \int_{0}^{\tau}(V_{\tau}-P_{t}%
)\mathrm{d}X_{t}\right)  \right\vert \mathcal{H}_{0}\right)  \label{2goal}%
\end{equation}
over the set of admissible strategies $X$ satisfying $\mathbf{(A1^{\prime})}$.

\begin{remark}
Observe that, in view of $\mathbf{(\mathbf{A1}^{\prime})}$, we have
$\mathrm{d}[Y,Y]_{t}=\sigma_{Z}^{2}(t)\mathrm{d}t.$
\end{remark}

We also have a general result in the case when $\tau$ is a predictable
stopping time for the insider. The same result is given in Corcuera et al. (2019).

\begin{proposition}
\label{Nefi} If $\tau$ is a predictable stopping time for the insider and $X$
is an optimal strategy then%
\[
V_{\tau}=P_{\tau}\text{ a.s.}%
\]

\end{proposition}

\begin{proof}
If the insider's strategy is such that $V_{\tau-}-P_{\tau-}\neq0$ then it is
suboptimal since the insider could approximate a jump at $\tau$ with the same
sign of $V_{\tau-}-P_{\tau-}$ by an absolutely continuous strategy and
improving her wealth

.
\end{proof}

\begin{remark}
From the economic point of view, due to Bertrand's type competition among
market makers, in the equilibrium market prices are \emph{rational}, or
\emph{competitive}, in the sense that
the competitive price is a price such that the expectation of the market
maker's profit equals zero. In fact, the total final wealth $W_{\tau}^{M}$ of
the market makers is given by
\[
W_{\tau}^{M}:=-Y_{\tau}\left(  V_{\tau}-P_{\tau}\right)  -\int_{0}^{\tau}%
Y_{t}\mathrm{d}P_{t},
\]
then, if $P_{t}=\mathbb{E}\left(  V_{t}|\mathcal{F}_{t}\right)  $, $0\leq
t\leq\tau$ see \eqref{competitive}, under the assumption that $\mathbb{E}%
\left(  \int_{0}^{\tau}Y_{t}^{2}\mathrm{d}[P,P]_{t}\right)  <\infty$, we have
that \textbf{ }$\mathbb{E}\left(  W_{\tau}^{M}\right)  =0.$
\end{remark}

\subsection{Main results}

First, we consider the risk-neutral case and reduce the admissibility
strategies to that fullfiling \
\[
\mathbf{(\mathbf{A1}^{\prime})}\text{, }\mathbb{E}\left(  \int_{0}^{T}\left(
P_{t}-V_{t}\right)  ^{2}\left(  \sigma_{Z}^{2}(t)+\sigma_{V}^{2}(t)\right)
\mathrm{d}t\right)  <\infty\text{ and }\mathbb{E}\left(  \int_{0}^{T}\left(
\bigtriangledown_{Y}P_{t}\right)  ^{2}\sigma_{Z}^{2}(t)\mathrm{d}t\right)
<\infty
\]

\begin{theorem}
\label{T1} Suppose that $\tau=T$ and that for all $t<T$, the price functional
$P$ is $\mathbb{C}_{b}^{1,3}$ and such that
\begin{equation}
\mathcal{D}_{t}P_{t}+\frac{1}{2}\bigtriangledown_{Y}^{2}P_{t}\sigma_{Z}%
^{2}(t)=0, \label{new}%
\end{equation}
with%
\begin{equation}
\mathcal{L}P_{t}:=\left[  \mathcal{D}_{t},\bigtriangledown_{Y}\right]  =0
\label{new2}%
\end{equation}
and
\begin{equation}
\bigtriangledown_{Y}P_{t}=G(t,P_{\cdot t}), \label{nas}%
\end{equation}
where $G(t,y_{\cdot t})\geq C>0$ is $\mathbb{C}_{b}^{1,2}$ and where we assume
that $\mathbb{E}\left(  \int_{0}^{T}G^{2}(t,P_{\cdot t})\sigma_{Z}%
^{2}(t)\mathrm{d}t\right)  <\infty$.

Then there is an equilibrium in the risk-neutral case if and only if
\[
(i)\text{ }P_{T}=V_{T},\qquad(ii)\text{ }Y\text{ is an }\mathbb{F}%
\text{-martingale, \ \ }(iii)\text{\ }G(t,y_{\cdot t})=g(t,y_{t})
\]
where $g(t,y)\geq C>0$ is $\mathcal{C}^{1,2}.$
\end{theorem}

\begin{proof}
Firstly, we prove that (i) and (ii) and (iii) are sufficient conditions. Set
the functional
\[
I(t,y_{\cdot t},,v):=\int_{v}^{y_{t}}\frac{z-v}{\bigtriangledown_{Y}%
P_{t}\left(  y_{\cdot t}^{z-y_{t}}\right)  }\mathrm{d}z
\]
Where
\[
\bigtriangledown_{Y}P_{t}\left(  y_{\cdot t}\right)  :=G(t,y_{\cdot t})
\]
and $y_{\cdot t}=P_{\cdot t}\left(  \omega\right)  ,$ for $0\leq t\leq T$.
Then we have
\begin{equation}
\bigtriangledown_{P}I(t,P_{\cdot t},,V_{t})=\frac{P_{t}-V_{t}}%
{\bigtriangledown_{Y}P_{t}\left(  P_{\cdot t}\right)  } \label{first}%
\end{equation}
and by the chain rule%
\begin{equation}
\bigtriangledown_{P}^{2}I(t,P_{\cdot t},,V_{t})=\frac{\bigtriangledown
_{Y}P_{t}\left(  P_{\cdot t}\right)  -\left(  P_{t}-V_{t}\right)
\bigtriangledown_{P}\left(  \bigtriangledown_{Y}P_{t}\left(  P_{\cdot
t}\right)  \right)  }{\bigtriangledown_{Y}P_{t}\left(  P_{\cdot t}\right)
^{2}}. \label{secon}%
\end{equation}
Consequently
\[
\left.  \left(  \bigtriangledown_{Y}P_{t}\left(  P_{\cdot t}^{z-P_{t}}\right)
-(z-V_{t})\partial_{z}\bigtriangledown_{Y}P_{t}\left(  P_{\cdot t}^{z-P_{t}%
}\right)  \right)  \right\vert _{z=P_{t}}=\bigtriangledown_{P}^{2}I(t,P_{\cdot
t},,V_{t})\left(  \bigtriangledown_{Y}P_{t}\left(  P_{\cdot t}\right)
\right)  ^{2},
\]
where
\[
\bigtriangledown_{Y}P_{t}\left(  P_{\cdot t}^{z-P_{t}}\right)  :=\left.
\bigtriangledown_{Y}P_{t}\left(  y_{\cdot t}^{z-y_{t}}\right)  \right\vert
_{y_{\cdot t}=P_{\cdot t}}%
\]
Then, we can write
\[
\bigtriangledown_{P}^{2}I(t,P_{\cdot t},,V_{t})\left(  \bigtriangledown
_{Y}P_{t}\left(  P_{\cdot t}\right)  \right)  ^{2}=-\int_{V_{t}}^{P_{t}%
}(z-V_{t})\partial_{z}^{2}\left(  \bigtriangledown_{Y}P_{t}\left(  P_{\cdot
t}^{z-P_{t}}\right)  \right)  \mathrm{d}z+\bigtriangledown_{Y}P_{t}\left(
P_{\cdot t}^{V_{t}-P_{t}}\right)  .
\]
On the other hand
\begin{equation}
\mathcal{D}_{t}I(t,P_{\cdot t},,V_{t})=-\int_{V_{t}}^{P_{t}}\frac{z-V_{t}%
}{\left(  \bigtriangledown_{Y}P_{t}\left(  P_{\cdot t}^{z-P_{t}}\right)
\right)  ^{2}}\mathcal{D}_{t}\left(  \bigtriangledown_{Y}P_{t}\left(  P_{\cdot
t}^{z-P_{t}}\right)  \right)  \mathrm{d}z \label{horP}%
\end{equation}
in such a way that,
\begin{align*}
&  \mathcal{D}_{t}I(t,P_{\cdot t},,V_{t})+\frac{1}{2}\bigtriangledown_{Y}%
^{2}I(t,P_{\cdot t},,V_{t})\left(  \bigtriangledown_{Y}P_{t}\left(  P_{\cdot
t}\right)  \right)  ^{2}\sigma_{Z}^{2}(t)\\
&  =-\int_{V_{t}}^{P_{t}}\left(  z-V_{t}\right)  \left(  \frac{\mathcal{D}%
_{t}\left(  \bigtriangledown_{Y}P_{t}\left(  P_{\cdot t}^{z-P_{t}}\right)
\right)  }{\left(  \bigtriangledown_{Y}P_{t}\left(  P_{\cdot t}^{z-P_{t}%
}\right)  \right)  ^{2}}+\frac{1}{2}\partial_{z}^{2}\left(  \bigtriangledown
_{Y}P_{t}\left(  P_{\cdot t}^{z-P_{t}}\right)  \right)  \sigma_{Z}%
^{2}(t)\right)  \mathrm{d}z\\
&  +\frac{1}{2}\bigtriangledown_{Y}P_{t}\left(  P_{\cdot t}^{V_{t}-P_{t}%
}\right)  \sigma^{2}(t).\\
&  =\int_{V_{t}}^{P_{t}}\left(  z-V_{t}\right)  \left(  \frac{\mathcal{D}%
_{t}\left(  \bigtriangledown_{Y}P_{t}\left(  P_{\cdot t}^{z-P_{t}}\right)
\right)  }{\left(  \bigtriangledown_{Y}P_{t}\left(  P_{\cdot t}^{z-P_{t}%
}\right)  \right)  ^{2}}+\frac{1}{2}\bigtriangledown_{P}^{2}\left(
\bigtriangledown_{Y}P_{t}\left(  P_{\cdot t}^{z-P_{t}}\right)  \right)
\sigma_{Z}^{2}(t)\right)  \mathrm{d}z\\
&  +\frac{1}{2}\bigtriangledown_{Y}P_{t}\left(  P_{\cdot t}^{V_{t}-P_{t}%
}\right)  \sigma^{2}(t).
\end{align*}
Notice that
\[
\partial_{z}\left(  \bigtriangledown_{Y}P_{t}\left(  P_{\cdot t}^{z-P_{t}%
}\right)  \right)  =\left.  \bigtriangledown_{P}\left(  \bigtriangledown
_{Y}P_{t}\left(  y_{\cdot t}\right)  \right)  \right\vert _{y_{\cdot
t}=P_{\cdot t}^{z-P_{t}}}%
\]
Now, conditions (\ref{new}) and (\ref{new2}) imply that
\[
\mathcal{D}_{t}\left(  \bigtriangledown_{Y}P\right)  +\frac{1}{2}\nabla
_{Y}^{2}\left(  \bigtriangledown_{Y}P_{t}\right)  \sigma_{Z}^{2}(t)=0.
\]
In this expression we are considering $\bigtriangledown_{Y}P$ as a functional
of $Y$, so to indicate that $\mathcal{D}_{t}$ is taken in this situation we
shall write $\mathcal{D}_{t}^{Y}$ and we use $\mathcal{D}_{t}$ to indicate
that $\mathcal{D}_{t}$ is taken \textit{freezing }$P.$ So we have
\[
\mathcal{D}_{t}^{Y}\left(  \bigtriangledown_{Y}P\right)  +\frac{1}{2}%
\nabla_{Y}^{2}\left(  \bigtriangledown_{Y}P_{t}\right)  \sigma_{Z}^{2}(t)=0
\]
Now if we consider $\bigtriangledown_{Y}P$ as a functional of $P,$ that is at
the same time a functional of $Y$, by applying the chain rule we have%
\begin{align*}
\nabla_{Y}^{2}\left(  \bigtriangledown_{Y}P_{t}\right)   &  =\nabla_{P}%
^{2}\left(  \bigtriangledown_{Y}P_{t}\right)  \left(  \bigtriangledown
_{Y}P_{t}\right)  ^{2}+\left(  \bigtriangledown_{P}\left(  \bigtriangledown
_{Y}P_{t}\right)  \right)  ^{2}\bigtriangledown_{Y}P_{t}\\
&  =\nabla_{P}^{2}\left(  \bigtriangledown_{Y}P_{t}\right)  \left(
\bigtriangledown_{Y}P_{t}\right)  ^{2}+\frac{\left(  \bigtriangledown_{Y}%
^{2}P_{t}\right)  ^{2}}{\bigtriangledown_{Y}P_{t}}%
\end{align*}%
\begin{align*}
\mathcal{D}_{t}^{Y}\left(  \bigtriangledown_{Y}P\right)   &  =\mathcal{D}%
_{t}\left(  \bigtriangledown_{Y}P\right)  +\bigtriangledown_{P}\left(
\bigtriangledown_{Y}P\right)  \mathcal{D}_{t}^{Y}P\\
&  =\mathcal{D}_{t}\left(  \bigtriangledown_{Y}P\right)  -\frac{1}{2}%
\frac{\bigtriangledown_{Y}^{2}P}{\bigtriangledown_{Y}P_{t}}\bigtriangledown
_{Y}^{2}P_{t}\sigma_{Z}^{2}(t),
\end{align*}
where in this last equality we use (\ref{new}). So we obtain that
\[
\mathcal{D}_{t}\left(  \bigtriangledown_{Y}P\right)  +\frac{1}{2}\nabla
_{P}^{2}\left(  \bigtriangledown_{Y}P_{t}\right)  \left(  \bigtriangledown
_{Y}P_{t}\right)  ^{2}\sigma_{Z}^{2}(t)=0.
\]
Note that we are considering functionals satisfying this condition for any
continuous trajectory with quadratic variation $\sigma_{Z}^{2}$, so by
continuity of the functionals we obtain that
\begin{align*}
&  \frac{\mathcal{D}_{t}\left(  \bigtriangledown_{Y}P_{t}\left(  P_{\cdot
t}^{z-P_{t}}\right)  \right)  }{\left(  \bigtriangledown_{Y}P_{t}\left(
P_{\cdot t}^{z-P_{t}}\right)  \right)  ^{2}}+\frac{1}{2}\sigma_{Z}%
^{2}(t)\bigtriangledown_{P}^{2}\left(  \bigtriangledown_{Y}P_{t}\left(
P_{\cdot t}^{z-P_{t}}\right)  \right) \\
&  =0.
\end{align*}
and
\begin{align}
&  \mathcal{D}_{t}I(t,P_{\cdot t},,V_{t})+\frac{1}{2}\bigtriangledown_{P}%
^{2}I(t,P_{\cdot t},,V_{t})\left(  \bigtriangledown_{Y}P_{t}\left(  P_{\cdot
t}\right)  \right)  ^{2}\sigma_{Z}^{2}(t)\nonumber\\
&  =\frac{1}{2}\bigtriangledown_{Y}P_{t}\left(  P_{\cdot t}^{V_{t}-P_{t}%
}\right)  \sigma_{Z}^{2}(t). \label{total}%
\end{align}
Finally if we apply the functional It\^o formula we have, since $V$ and $Z$ are
independent and we consider only absolutely continuous strategies, see
$\mathbf{(\mathbf{A1}^{\prime}),}$
\begin{align*}
I(T,P_{\cdot T},V_{T})  &  =I(0,P_{0},V_{0})+\int_{0}^{T}\mathcal{D}%
_{t}I(t,P_{\cdot t},,V_{t})\mathrm{d}t+\int_{0}^{T}\bigtriangledown
_{P}I(t,P_{\cdot t},,V_{t})\mathrm{d}P_{t}+\int_{0}^{T}\bigtriangledown
_{V}I(t,P_{\cdot t},,V_{t})\mathrm{d}V_{t}\\
&  +\frac{1}{2}\int_{0}^{T}\bigtriangledown_{P}^{2}I(t,P_{\cdot t}%
,,V_{t})\left(  \bigtriangledown_{Y}P_{t}\left(  P_{\cdot t}\right)  \right)
^{2}\sigma_{Z}^{2}(t)\mathrm{d}t+\frac{1}{2}\int_{0}^{T}\bigtriangledown
_{V}^{2}I(t,P_{\cdot t},,V_{t})\sigma_{V}^{2}(t)\mathrm{d}t,
\end{align*}
Also, we have use the fact that
\[
\mathrm{d}[P,P]_{t}=\left(  \bigtriangledown_{Y}P_{t}\left(  P_{\cdot
t}\right)  \right)  ^{2}\mathrm{d}[Y,Y]_{t}=\left(  \bigtriangledown_{Y}%
P_{t}\left(  P_{\cdot t}\right)  \right)  ^{2}\sigma_{Z}^{2}(t)\mathrm{d}t.
\]
Then by \ (\ref{new}) (\ref{first}) and (\ref{total})
\begin{align*}
I(T,P_{\cdot T},V_{T})  &  =I(0,P_{0},V_{0})+\int_{0}^{T}(P_{t}-V_{t}%
)\mathrm{d}Y_{t}+\int_{0}^{T}\bigtriangledown_{V}I(t,P_{\cdot t}%
,,V_{t})\mathrm{d}V_{t}\\
&  +\frac{1}{2}\int_{0}^{T}\bigtriangledown_{Y}P_{t}\left(  P_{\cdot t}%
^{V_{t}-P_{t}}\right)  \sigma^{2}(t)\mathrm{d}t+\frac{1}{2}\int_{0}%
^{T}\bigtriangledown_{V}^{2}I(t,P_{\cdot t},,V_{t})\sigma_{V}^{2}%
(t)\mathrm{d}t,
\end{align*}
and
\begin{align}
&  \int_{0}^{T}\left(  V_{t}-P_{t}\right)  \mathrm{d}X_{t}-\left(
I(0,P_{0},V_{0})+\frac{1}{2}\int_{0}^{T}\bigtriangledown_{V}^{2}I(t,P_{\cdot
t},,V_{t})\sigma_{V}^{2}(t)\mathrm{d}t+\frac{1}{2}\int_{0}^{T}\bigtriangledown
_{Y}P_{t}\left(  P_{\cdot t}^{V_{t}-P_{t}}\right)  \sigma_{Z}^{2}%
(t)\mathrm{d}t\right) \nonumber\\
&  =-I(T,P_{\cdot T},V_{T})+\int_{0}^{T}\left(  P_{t}-V_{t}\right)
\mathrm{d}Z_{t}+\int_{0}^{T}\bigtriangledown_{V}I(t,P_{\cdot t},,V_{t}%
)\mathrm{d}V_{t}. \label{maxwealth}%
\end{align}
We have that
\[
\left\vert \bigtriangledown_{V}I(t,P_{\cdot t},,V_{t})\right\vert =\left\vert
\int_{V_{t}}^{P_{t}}\frac{-1}{\bigtriangledown_{Y}P_{t}\left(  P_{\cdot
t}^{z-P_{t}}\right)  }\mathrm{d}z\right\vert <\frac{\left\vert P_{t}%
-V_{t}\right\vert }{C},
\]
therefore%
\[
\mathbb{E}\left(  \int_{0}^{T}\left(  \bigtriangledown_{V}I(t,P_{\cdot
t},,V_{t})\right)  ^{2}\sigma_{V}^{2}(t)\mathrm{d}t\right)  <\frac{1}{C^{2}%
}\mathbb{E}\left(  \int_{0}^{T}\left(  P_{t}-V_{t}\right)  ^{2}\sigma_{V}%
^{2}(t)\mathrm{d}t\right)  <\infty,
\]
and consequently $\mathbb{E}\left(  \int_{0}^{T}\bigtriangledown
_{V}I(t,P_{\cdot t},,V_{t})\mathrm{d}V_{t}\right)  =0$. Also $\mathbb{E}%
\left(  \int_{0}^{T}\left(  P_{t}-V_{t}\right)  \mathrm{d}Z_{t}\right)  =0$
since $\mathbb{E}\left(  \int_{0}^{T}\left(  P_{t}-V_{t}\right)  ^{2}%
\sigma_{Z}^{2}(t)\mathrm{d}t\right)  <\infty$. If $(iii)$
\[
G(t,y_{\cdot t})=g(t,y_{t})
\]%
\[
\bigtriangledown_{Y}P_{t}\left(  P_{\cdot t}^{V_{t}-P_{t}}\right)
=g(t,V_{t})
\]
Finally,
\[
\bigtriangledown_{V}^{2}I(t,P_{\cdot t},,V_{t})=\frac{1}{\bigtriangledown
_{Y}P_{t}\left(  P_{\cdot t}^{V_{t}-P_{t}}\right)  }=\frac{1}{g(t,V_{t}%
)}<\frac{1}{C}%
\]
in a way that
\[
\frac{1}{2}\int_{0}^{T}\bigtriangledown_{V}^{2}I(t,P_{\cdot t},,V_{t}%
)\sigma_{V}^{2}(t)\mathrm{d}t,\text{ and }\frac{1}{2}\int_{0}^{T}%
\bigtriangledown_{Y}P_{t}\left(  P_{\cdot t}^{V_{t}-P_{t}}\right)  \sigma
_{Z}^{2}(t)\mathrm{d}t
\]
only depend on $V.$ From this, it is easy to see that
\[
I(t,P_{\cdot t},,V_{t})=I(t,P_{t},V_{t})
\]
By (\ref{first}) and $(i)$, we have
\[
\partial_{2}I(T,P_{T},V_{T})=\frac{P_{T}-V_{T}}{g(T,P_{T})}=0
\]
and, by (\ref{secon}) and $(i)$, we obtain
\[
\partial_{22}I(T,P_{T},V_{T})=\frac{1}{g(T,P_{T})}-\frac{\left(  P_{T}%
-V_{T}\right)  \partial_{2}g(T,P_{T})}{g^{2}(T,P_{T})}=\frac{1}{g(T,P_{T}%
)}>0.
\]
So we have a maximum of $-\mathbb{E}\left(  I(T,P_{T},V_{T})\right)  .$ We
also have that $P$ is an $\mathbb{F}$-martingale from (\ref{new}), the
integrability condition and $(ii)$. Therefore by $(i)$ and since $V$ is an
$\mathbb{H}$-martingale, we obtain that
\[
P_{t}=\mathbb{E}\left(  P_{T}|\mathcal{F}_{t}\right)  =\mathbb{E}\left(
V_{T}|\mathcal{F}_{t}\right)  =\mathbb{E}\left(  \mathbb{E}\left(
V_{T}|\mathcal{H}_{t}\right)  |\mathcal{F}_{t}\right)  =\mathbb{E}\left(
V_{t}|\mathcal{F}_{t}\right)  .
\]
Now we show that $(i)$ and $(ii)$ are necessary conditions. In fact $(i)$ is
necessary by Proposition \ref{Nefi}. By (\ref{new}) and the functional Itô's
formula, we have
\[
\mathrm{d}P_{t}=\bigtriangledown_{Y}P_{t}\mathrm{d}Y_{t},
\]
then the result follows from the fact that $P$ is an $\mathbb{F}$-martingale
and $\bigtriangledown_{Y}P_{t}\geq C>0$.

Condition $(iii)$ is also necessary if we want to maximize $\mathbb{E}\left(
\left.  \int_{0}^{T}\left(  V_{t}-P_{t}\right)  \mathrm{d}X_{t}\right\vert
\mathcal{H}_{0}\right)  $. In fact, by (\ref{maxwealth})
\begin{align}
&  \int_{0}^{T}\left(  V_{t}-P_{t}\right)  \mathrm{d}X_{t}-I(0,P_{0}%
,V_{0})\nonumber\\
&  =-I(T,P_{\cdot T},V_{T})+\frac{1}{2}\int_{0}^{T}J_{t}(t,P_{\cdot t}%
,,V_{t})\mathrm{d}t+\int_{0}^{T}\left(  P_{t}-V_{t}\right)  \mathrm{d}%
Z_{t}+\int_{0}^{T}\bigtriangledown_{V}I(t,P_{\cdot t},,V_{t})\mathrm{d}V_{t}.
\end{align}
with
\begin{align*}
J_{t}(t,P_{\cdot t},,V_{t})  &  :=\bigtriangledown_{V}^{2}I(t,P_{\cdot
t},,V_{t})\sigma_{V}^{2}(t)+\bigtriangledown_{Y}P_{t}\left(  P_{\cdot
t}^{V_{t}-P_{t}}\right)  \sigma_{Z}^{2}(t)\\
&  =\frac{1}{\bigtriangledown_{Y}P_{t}\left(  P_{\cdot t}^{V_{t}-P_{t}%
}\right)  }\sigma_{V}^{2}(t)+\bigtriangledown_{Y}P_{t}\left(  P_{\cdot
t}^{V_{t}-P_{t}}\right)  \sigma_{Z}^{2}(t).
\end{align*}
then, for fixed $\omega,$ $J_{t}$ is an unbounded convex function of
$\bigtriangledown_{Y}P_{t}\left(  P_{\cdot t}^{V_{t}-P_{t}}\right)  .$
Therefore we can modify the strategy in order to get $\bigtriangledown
_{Y}P_{t}\left(  P_{\cdot t}^{V_{t}-P_{t}}\right)  $ as large (or small if
$\sigma_{V}^{2}(t)\neq0$) as we want and at the same time keeping $P_{T}%
=V_{T}.$ That shows that the optimal wealth is not bounded and there is not
equilibrium in such a situation except if $\bigtriangledown_{Y}P_{t}(y_{\cdot
t})=g(t,y_{t})$, that it is a function of the spot value. In other words if
\[
G(t,y_{\cdot t})=g(t,y_{t})
\]
where $g(t,y)\geq C>0$ is $\mathcal{C}^{1,2}.$
\end{proof}

We can obtain an analogous result to Theorem \ref{T1} for the non-risk-neutral
case when the utility function is $U(x)=\gamma e^{\gamma x},\gamma<0$ and when
$V_{t}\equiv V.$

We reduce the admissibility strategies to that fullfilling
\[
\mathbf{(\mathbf{A1}^{\prime})}\text{, }\mathbb{E}\left(  \exp\left\{
\frac{1}{2}\gamma^{2}\int_{0}^{T}\left(  P_{t}-V\right)  ^{2}\sigma_{Z}%
^{2}(t)\mathrm{d}t\right\}  \right)  \text{ and }\mathbb{E}\left(  \int%
_{0}^{T}G^{2}(t,P_{t})\sigma_{Z}^{2}(t)\mathrm{d}t\right)  <\infty
\]

\begin{theorem}
Let $V_{t}\equiv V$ and $\tau=T.$ For any $t<T$, let $P\in\mathbb{C}_{b}%
^{1,3}$ be a price functional such that
\begin{equation}
\mathcal{D}_{t}P_{t}+\frac{1}{2}\bigtriangledown_{Y}^{2}P_{t}\sigma_{Z}%
^{2}(t)=0, \label{NRN1}%
\end{equation}
with%
\begin{equation}
\mathcal{L}P_{t}:=\mathcal{D}_{t}\bigtriangledown_{Y}P_{t}-\bigtriangledown
_{Y}\mathcal{D}_{t}P_{t}=\gamma\sigma_{Z}^{2}(t)\left(  \bigtriangledown
_{Y}P_{t}\right)  ^{2} \label{NRN2}%
\end{equation}
and
\[
\bigtriangledown_{Y}P_{t}=G(t,P_{\cdot t}),
\]
where where $G(t,y_{\cdot t})\geq C>0$ is $\mathbb{C}^{1,2}$.

Then there is an equilibrium in the non-risk-neutral case, with utility
function $U(x)=\gamma e^{\gamma x}$, if and only if
\[
(i)\text{ }P_{T}=V,\qquad(ii)\text{ }Y\text{ is an }\mathbb{F}%
\text{-martingale \ \ \ \ }(iii)\text{ }G(t,y_{\cdot t})=g(t,y_{t})
\]
where $g(t,y)\geq C>0$ is $\mathcal{C}^{1,2}.$
\end{theorem}

\begin{proof}
As in the previous proof, let
\[
I(t,y_{\cdot t},,v):=\int_{v}^{y_{t}}\frac{z-v}{\bigtriangledown_{Y}%
P_{t}\left(  y_{\cdot t}^{z-y_{t}}\right)  }\mathrm{d}z,
\]
then, by the functional Itô formula,%
\begin{align*}
I(T,P_{\cdot T},V)  &  =I(0,P_{0},V)+\int_{0}^{T}\mathcal{D}_{t}I(t,P_{\cdot
t},,V)\mathrm{d}t+\int_{0}^{T}\bigtriangledown_{P}I(t,P_{\cdot t}%
,,V_{t})\mathrm{d}P_{t}\\
&  +\frac{1}{2}\int_{0}^{T}\bigtriangledown_{P}^{2}I(t,P_{\cdot t},,V)\left(
\bigtriangledown_{Y}P_{t}\left(  P_{\cdot t}\right)  \right)  ^{2}\sigma
_{Z}^{2}(t)\mathrm{d}t,
\end{align*}
and by (\ref{NRN1}) and (\ref{NRN2})
\[
\mathcal{D}_{t}\left(  \bigtriangledown_{Y}P\right)  +\frac{1}{2}\nabla
_{P}^{2}\left(  \bigtriangledown_{Y}P_{t}\right)  \left(  \bigtriangledown
_{Y}P_{t}\right)  ^{2}\sigma_{Z}^{2}(t)=\gamma\sigma_{Z}^{2}(t).
\]
Analougously to the previous theorem we have
\begin{align}
&  \mathcal{D}_{t}I(t,P_{\cdot t},,V_{t})+\frac{1}{2}\bigtriangledown_{P}%
^{2}I(t,P_{\cdot t},,V_{t})\left(  \bigtriangledown_{Y}P_{t}\left(  P_{\cdot
t}\right)  \right)  ^{2}\sigma_{Z}^{2}(t)\nonumber\\
&  =-\gamma\int_{V}^{P_{t}}(z-v)\mathrm{d}z+\bigtriangledown_{Y}P_{t}\left(
P_{\cdot t}^{V-P_{t}}\right)  \sigma_{Z}^{2}(t)\\
&  =-\frac{\gamma}{2}\left(  P_{t}-V\right)  ^{2}\sigma_{Z}^{2}%
(t)+\bigtriangledown_{Y}P_{t}\left(  P_{\cdot t}^{V-P_{t}}\right)  \sigma
_{Z}^{2}(t)
\end{align}
and
\begin{align*}
&  \int_{0}^{T}\left(  V-P_{\cdot t}\right)  \mathrm{d}X_{t}-\left(
I(0,P_{0},V)+\int_{0}^{T}\bigtriangledown_{Y}P_{t}\left(  P_{\cdot t}%
^{V-P_{t}}\right)  \sigma_{Z}^{2}(t)\mathrm{d}t\right) \\
&  =-I(T,P_{\cdot T},V)+\int_{0}^{T}\left(  P_{t}-V\right)  \mathrm{d}%
Z_{t}-\frac{1}{2}\gamma\int_{0}^{T}\left(  P_{t}-V\right)  ^{2}\sigma_{Z}%
^{2}(t)\mathrm{d}t.
\end{align*}
Therefore,
\begin{align*}
&  \gamma\exp\left\{  \gamma\int_{0}^{T}\left(  V-P_{\cdot t}\right)
\mathrm{d}X_{t}\right\}  \exp\left\{  -\gamma I(0,P_{0},V)-\gamma\int_{0}%
^{T}\bigtriangledown_{Y}P_{t}\left(  P_{\cdot t}^{V-P_{t}}\right)  \sigma
_{Z}^{2}(t)\mathrm{d}t\right\} \\
&  =\gamma\exp\left\{  -\gamma I(T,P_{\cdot T},V)\right\}  \exp\left\{
\gamma\int_{0}^{T}\left(  P_{t}-V\right)  \mathrm{d}Z_{t}-\frac{1}{2}%
\gamma^{2}\int_{0}^{T}\left(  P_{t}-V\right)  ^{2}\sigma_{Z}^{2}%
(t)\mathrm{d}t\right\}  .
\end{align*}
Now if $(iii)$ it can be seen that
\[
I(T,P_{\cdot T},V)=I(T,P_{T},V)
\]
and consequently
\[
\partial_{2}I(T,P_{T},V_{T})=\frac{P_{T}-V_{T}}{g(T,P_{T})}=0
\]%
\[
\partial_{22}I(T,P_{T},V)=\frac{1}{g(T,P_{T})}-\left(  P_{T}-V\right)
\partial_{2}g(T,P_{T})=\frac{1}{g(T,P_{T})}>0.
\]
So the minimum value of $I(T,P_{T},V)$ is when $P_{T}=V$ and its value is
$I(T,P_{T},V):=\int_{V}^{P_{T}}\frac{z-V}{g(t,z)}\mathrm{d}z=0.$ Then, since
$\gamma<0$,
\begin{align*}
&  \mathbb{E}\left(  \gamma\exp\left\{  \gamma\int_{0}^{T}\left(
V-P_{t}\right)  \mathrm{d}X_{t}\right\}  \exp\left\{  -\gamma I(0,P_{0}%
,V)-\gamma\int_{0}^{T}g(t,V)\mathrm{d}t\right\}  \right) \\
&  \leq\gamma\mathbb{E}\left(  \exp\left\{  \gamma\int_{0}^{T}\left(
P_{t}-V\right)  \mathrm{d}Z_{t}-\frac{1}{2}\gamma^{2}\int_{0}^{T}\left(
P_{t}-V\right)  ^{2}\sigma_{Z}^{2}(t)\mathrm{d}t\right\}  \right)  =\gamma.
\end{align*}
And we get the maximum value of $\mathbb{E}\left(  \gamma\exp\left\{
\gamma\int_{0}^{T}\left(  V-P_{t}\right)  \mathrm{d}X_{t}\right\}  \right)  $
when $P_{T}=V.$ The rest of the proof is analogous to the one of the previous theorem.
\end{proof}

\section{Necessary conditions for the equilibrium pricing rules}

In this section we study general necessary conditions to obtain an equilibrium
and we see that the classes of price functionals of the previous section,
characterised by the relationships (\ref{new}) and (\ref{new2}) for the
risk-neutral insider and (\ref{NRN1}) and (\ref{NRN2}) for the risk-averse
one, are actually justified by the arguments that follow. Note that in this
section, the release time of information $\tau$ is assumed predictable and
bounded. A remark at the end of the session deals with the case of $\tau$
independent of the observable variables.


Here below we study the effect of an $\varepsilon$-perturbation of the insider
strategies:
\[
\mathrm{d}X_{t}^{\left(  \varepsilon\right)  }:=\mathrm{d}X_{t}+\varepsilon
\beta_{t}\mathrm{d}t,
\]
where $\beta$ is a bounded adapted processes, in the prices $P_{t}%
=P_{t}(Z_{\cdot t}+X_{\cdot t}).$

From now on, we are going to assume that there exist a \emph{strictly
positive} $\mathcal{B(}\mathbb{R}_{+})\otimes\mathcal{P}^{{\mathbb{F}}}%
$-measurable\footnote{$\mathcal{P}^{\mathbb{F}}$ denotes the $\mathbb{F}%
$-predictable $\sigma$-field.} function $K(s,t)(\omega)$, $0\leq s\leq
t\leq\tau$, $\omega\in\Omega$, continuous for all $0\leq s\leq t\leq\tau$,
such that, for a.a. $t$,
\[
\mathbf{(R)}\text{ \ }\qquad P_{t}^{(\varepsilon)}-P_{t}=\varepsilon\int%
_{0}^{t}K(s,t)\beta_{s}\mathrm{d}s+o\left(  \varepsilon\right)  R_{t},
\]
when we make an $\varepsilon$-perturbation of the strategies. Here above
$P_{t}^{(\varepsilon)}:=P_{t}(Z_{\cdot t}+X_{\cdot t}^{(\varepsilon)}),$ and
$R$ is a bounded progressively measurable process, independent of $\beta$.
Observe that the random variables $K(s,t)$ are strictly positive because
$P_{t}=P_{t}(Y_{\cdot t})$ is a strictly increasing functional. Note that, as
a consequence of $\mathbf{(R)}$, we have that
\[
\lim_{\varepsilon\rightarrow0}\frac{P_{t}^{(\varepsilon)}-P_{t}}{\varepsilon
}=\int_{0}^{t}K(s,t)\beta_{s}\mathrm{d}s.
\]

\begin{proposition}
\label{KV} Assume that $P$ \ is continuous for fixed times and that, for any
bounded adapted process $\beta$, $\mathbf{(R)}$ holds by means of the kernels
$K$ described above. Then
\[
\nabla_{Y}P_{t}=K(t,t).
\]

\end{proposition}

\begin{proof}
Set, for fixed $t$ and with $r<t$,
\[
\beta_{s}^{(r)}:=\frac{1}{t-r}\mathbf{1}_{[r,t]}(s).
\]
Taking limits in $\mathbf{(R)}$ when $r\rightarrow t$ we have that, a.s.
$\mathbb{P\otimes}Leb,$
\[
P_{t}(Y_{\cdot t}^{(\varepsilon)})-P_{t}(Y_{\cdot t})=\varepsilon
K(t,t)+o\left(  \varepsilon\right)  R_{t}%
\]
By this we can conclude.
\end{proof}

The next result presents a factorisation property of the kernel and a
sufficient condition to obtain it.

\begin{proposition}
\label{factor}Let $G$ and $F$ be $\mathcal{C}^{1,2}$. Assume that
\begin{equation}
\nabla_{Y}P_{t}=G(t,P_{t}), \label{CV}%
\end{equation}
and
\begin{equation}
\mathcal{D}_{t}P_{t}+\frac{1}{2}\bigtriangledown_{Y}^{2}P_{t}\sigma_{Z}%
^{2}(t)=F(t,P_{t}) \label{CO}%
\end{equation}
hold. Then the kernel $K$ admits factorisation
\begin{equation}
K(s,t)=K_{1}(s)K_{2}(t), \label{fac}%
\end{equation}
with
\[
K_{2}(t)=\mathcal{E}\left(  \int_{0}^{t}\partial_{2}G(s,P_{s})\mathrm{d}%
Y_{s}\right)  \exp\left(  \int_{0}^{t}\partial_{2}F(s,P_{s})\mathrm{d}%
s\right)  ,
\]
where $\mathcal{E}$ is the stochastic exponential, and
\[
K_{1}(t)=\frac{G(t,P_{t})}{K_{2}(t)}.
\]
Moreover $\left[  K_{1},K_{1}\right]  \equiv0.$
\end{proposition}

\begin{proof}
Since
\[
P_{t}=P_{0}+\int_{0}^{t}\nabla_{Y}P_{s}\mathrm{d}Y_{s}+\int_{0}^{t}\left(
\mathcal{D}_{s}P_{s}+\frac{1}{2}\bigtriangledown_{Y}^{2}P_{s}\sigma_{Z}%
^{2}(s)\right)  \mathrm{d}s
\]
we have that$,$%
\begin{equation}
\left.  \frac{\mathrm{d}P_{t}^{(\varepsilon)}}{\mathrm{d}\varepsilon
}\right\vert _{\varepsilon=0}=\int_{0}^{t}\partial_{2}G\left.  \frac
{\mathrm{d}P_{s}^{(\varepsilon)}}{\mathrm{d}\varepsilon}\right\vert
_{\varepsilon=0}\mathrm{d}Y_{s}+\int_{0}^{t}G\beta_{s}\mathrm{d}s+\int_{0}%
^{t}\partial_{2}F\left.  \frac{\mathrm{d}P_{s}^{(\varepsilon)}}{\mathrm{d}%
\varepsilon}\right\vert _{\varepsilon=0}\mathrm{d}s. \label{dp1}%
\end{equation}
Therefore
\begin{align}
\left.  \frac{\mathrm{d}P_{t}^{(\varepsilon)}}{\mathrm{d}\varepsilon
}\right\vert _{\varepsilon=0}  &  =\mathcal{E}\left(  \int_{0}^{t}\partial
_{2}G(s,P_{s})\mathrm{d}Y_{s}\right)  \exp\left(  \int_{0}^{t}\partial
_{2}F(s,P_{s})\mathrm{d}s\right) \nonumber\\
&  \times\int_{0}^{t}\frac{G(s,P_{s})\beta_{s}}{\mathcal{E}\left(  \int%
_{0}^{s}\partial_{2}G(u,P_{u})\mathrm{d}Y_{u}\right)  \exp\left(  \int_{0}%
^{s}\partial_{2}F(u,P_{u})\mathrm{d}u\right)  }\mathrm{d}s. \label{dp2}%
\end{align}
This is easy to be verified by showing that the differentials and the values
at $t=0$ of $\left.  \frac{\mathrm{d}P_{t}^{(\varepsilon)}}{\mathrm{d}%
\varepsilon}\right\vert _{\varepsilon=0}$in (\ref{dp1}) and (\ref{dp2}) are
the same. Finally, by a uniqueness argument, we have that
\[
\left.  \frac{\mathrm{d}P_{t}^{(\varepsilon)}}{\mathrm{d}\varepsilon
}\right\vert _{\varepsilon=0}=\int_{0}^{t}K_{1}(s)K_{2}(t)\beta_{s}%
\mathrm{d}s,
\]
with%
\[
K_{2}(t)=\mathcal{E}\left(  \int_{0}^{t}\partial_{2}G(s,P_{s})\mathrm{d}%
Y_{s}\right)  \exp\left(  \int_{0}^{t}\partial_{2}F(s,P_{s})\mathrm{d}%
s\right)  ,
\]
and
\[
K_{1}(t)=\frac{G(t,P_{t})}{\mathcal{E}\left(  \int_{0}^{t}\partial
_{2}G(u,P_{u})\mathrm{d}Y_{u}\right)  \exp\left(  \int_{0}^{t}\partial
_{2}F(u,P_{u})\right)  }.
\]
Finally it is easy to see that
\begin{equation}
\mathrm{d}K_{1}(t)=\frac{\partial_{1}G+\frac{1}{2}G^{2}\partial_{22}%
G\sigma_{Z}^{2}(t)}{K_{2}(t)}\mathrm{d}t+\partial_{2}G\frac{\mathcal{D}%
_{t}P_{t}+\frac{1}{2}\bigtriangledown_{Y}^{2}P_{t}}{K_{2}(t)}\mathrm{d}%
t-\frac{G\partial_{2}F}{K_{2}(t)}\mathrm{d}t. \label{dk1}%
\end{equation}
\ 
\end{proof}

In particular we obtain the following

\begin{proposition}
\label{Bracket}Let $P$ be a price functional such that \eqref{CV} holds and
\eqref{CO} holds for $F\equiv0$, i.e.
\begin{equation}
\mathcal{D}_{t}P_{t}+\frac{1}{2}\bigtriangledown_{Y}^{2}P_{t}\sigma_{Z}%
^{2}(t)=0. \label{dP}%
\end{equation}
Then
\[
\mathcal{L}P_{t}=K_{2}(t)\frac{\mathrm{d}}{\mathrm{d}t}K_{1}(t).
\]

\end{proposition}

\begin{proof}
By (\ref{dP}) we have
\begin{align*}
\mathcal{L}P_{t}  &  =\mathcal{D}_{t}\bigtriangledown_{Y}P_{t}+\frac{1}%
{2}\bigtriangledown_{Y}\left(  \bigtriangledown_{Y}^{2}P_{t}\right)
\sigma_{Z}^{2}(t)\\
&  =\partial_{1}G+\frac{1}{2}G^{2}\partial_{22}G\sigma_{Z}^{2}(t)+\frac{1}%
{2}\partial_{2}G\left(  \mathcal{D}_{t}P_{t}+\frac{1}{2}\bigtriangledown
_{Y}^{2}P_{t}\sigma_{Z}^{2}(t)\right) \\
&  =\partial_{1}G+\frac{1}{2}G^{2}\partial_{22}G\sigma_{Z}^{2}(t).
\end{align*}
Now by (\ref{dk1}) and since $F\equiv0$, we have that%
\[
\mathcal{L}P_{t}=K_{2}(t)\frac{\mathrm{d}K_{1}(t)}{\mathrm{d}t}.
\]

\end{proof}

We have obtain a general result with a necessary condition for the an optimal strategy.

\begin{theorem}
\label{T4} Assume that for all $\beta$ bounded $(\mathbf{R})$ holds in terms
of the kernel $K$ as above. If $X$ is optimal, then we have
\begin{equation}
\left.  \mathbf{1}_{[0,\tau)}(t)\mathbb{E}\left(  \left.  U^{\prime}(W_{\tau
})\left(  V_{\tau}-P_{t}\right)  \right\vert \mathcal{H}_{t}\right)
-\mathbb{E}\left(  \left.  \int_{t\wedge\tau}^{\tau}\mathbb{E}\left(  \left.
U^{\prime}(W_{\tau})\right\vert \mathcal{H}_{s}\right)  K(t,s)\mathrm{d}%
X_{s}\right\vert \mathcal{H}_{t}\right)  =0,\text{ a.s.-}\mathbb{P\otimes
}Leb\text{ }\right.  \text{ } \label{NC}%
\end{equation}

\end{theorem}

\begin{proof}
Take $\mathrm{d}X_{t}^{\left(  \varepsilon\right)  }:=\mathrm{d}%
X_{t}+\varepsilon\beta_{t}\mathrm{d}t,$ \ where $\beta$ is a bounded adapted
processes, then,
\begin{align*}
&  \mathbb{E}\left(  U(W_{\tau}^{(\varepsilon)})-U(W_{\tau})\right) \\
&  =\mathbb{E}\left(  U\left(  \int_{0}^{\tau}\left(  V_{\tau}-P_{t}%
^{(\varepsilon)}\right)  \mathrm{d}X_{t}^{(\varepsilon)}\right)  -U\left(
W_{\tau}\right)  )\right) \\
&  =\varepsilon\mathbb{E}\left(  U^{\prime}(W_{\tau})\left(  \int_{0}^{\tau
}\left(  V_{\tau}-P_{t}\right)  \beta_{t}\mathrm{d}t-\int_{0}^{\tau}\left(
\int_{0}^{t}K(s,t)\beta_{s}\mathrm{d}s\right)  \mathrm{d}X_{t}\right)
\right)  +o(\varepsilon)\\
&  =\varepsilon\mathbb{E}\left(  U^{\prime}(W_{\tau})\left(  \int_{0}^{\tau
}\left(  V_{\tau}-P_{t}-\int_{t}^{\tau}K(t,s)\mathrm{d}X_{s}\right)  \beta
_{t}\mathrm{d}t\right)  \right)  +o(\varepsilon).
\end{align*}
Note that, by Fubini's theorem,
\[
\int_{0}^{\tau}\left(  \int_{0}^{t}K(s,t)\beta_{s}\mathrm{d}s\right)
\mathrm{d}X_{t}=\int_{0}^{\tau}\left(  \int_{t}^{\tau}K(t,s)\mathrm{d}%
X_{s}\right)  \beta_{t}\mathrm{d}t.
\]
Then
\[
\left.  \frac{\mathrm{d}\mathbb{E}\left(  U\left(  W_{\tau}^{(\varepsilon
)}\right)  \right)  }{\mathrm{d}\varepsilon}\right\vert _{\varepsilon=0}=0
\]
implies that
\[
\mathbb{E}\left(  \int_{0}^{\tau}U^{\prime}(W_{\tau})\left(  V_{\tau}%
-P_{t}-\int_{t}^{\tau}K(t,s)\mathrm{d}X_{s}\right)  \beta_{t}\mathrm{d}%
t\right)  =0.
\]
Since we can take $\beta_{t}=\alpha_{u}\mathbf{1}_{(u,u+h]}(t),$ with
$\alpha_{u}$ measurable and bounded and $\tau$ is a stopping time, we have
that%
\[
\mathbf{1}_{[0,\tau)}(t)\mathbb{E}\left(  \left.  U^{\prime}(W_{\tau})\left(
V_{\tau}-P_{t}-\int_{t}^{\tau}K(t,s)\mathrm{d}X_{s}\right)  \right\vert
\mathcal{H}_{t}\right)  =0,
\]
a.s.-$\mathbb{P\otimes}Leb$. And, from the Law of Iterated Expectations
\[
\mathbf{1}_{[0,\tau)}(t)\mathbb{E}\left(  \left.  U^{\prime}(W_{\tau})\left(
V_{\tau}-P_{t}\right)  \right\vert \mathcal{H}_{t}\right)  -\mathbb{E}\left(
\left.  \int_{t}^{\tau}\mathbb{E}\left(  \left.  U^{\prime}(W_{\tau
})\right\vert \mathcal{H}_{s}\right)  K(t,s)\mathrm{d}X_{s}\right\vert
\mathcal{H}_{t}\right)  =0
\]

\end{proof}

The result above allows us to give some necessary conditions for an equilibrium.

\begin{proposition}
\label{neccond} In the conditions of the Theorem \ref{T4} and assuming that
(\ref{CV}) and (\ref{CO}) hold, we have that if $(P,X)$ is an equilibrium,
then
\begin{align}
0  &  =\mathbb{E}\left(  \left.  U^{\prime}(W_{\tau})\left(  V_{\tau}%
-P_{t}\right)  \right\vert \mathcal{H}_{t}\right)  \frac{\mathrm{d}%
}{\mathrm{d}t}\left(  \frac{1}{K_{1}(t)}\right)  -\frac{\mathbb{E}\left(
\left.  U^{\prime}(W_{\tau})\right\vert \mathcal{H}_{t}\right)  }{K_{1}%
(t)}\left(  \mathcal{D}_{t}P_{t}+\frac{1}{2}\bigtriangledown^{2}P_{t}%
\sigma_{Z}^{2}(t)\right) \nonumber\\
&  -\frac{1}{K_{1}(\cdot)}\frac{\mathrm{d}}{\mathrm{d}t}\left[  P,\mathbb{E}%
\left(  \left.  U^{\prime}(W_{\tau})\right\vert \mathcal{H}_{\cdot}\right)
\right]  _{0}^{t}. \label{main}%
\end{align}

\end{proposition}

\begin{proof}
Thanks to $\mathbf{(\mathbf{A1}^{\prime}),(R),}$ the factorisation property
\eqref{fac}, and by means of Theorem \ref{FIT} and Proposition \ref{KV}, we
have that
\begin{align*}
&  \mathbb{E}\left(  \left.  \int_{t}^{\tau}\mathbb{E}\left(  \left.
U^{\prime}(W_{\tau})\right\vert \mathcal{H}_{s}\right)  K(t,s)\mathrm{d}%
X_{s}\right\vert \mathcal{H}_{t}\right) \\
&  =K_{1}(t)\mathbb{E}\left(  \left.  \int_{t}^{\tau}\frac{1}{K_{1}%
(s)}\mathbb{E}\left(  \left.  U^{\prime}(W_{\tau})\right\vert \mathcal{H}%
_{s}\right)  K(s,s)\mathrm{d}Y_{s}\right\vert \mathcal{H}_{t}\right) \\
&  =K_{1}(t)\mathbb{E}\left(  \left.  \int_{t}^{\tau}\frac{1}{K_{1}%
(s)}\mathbb{E}\left(  \left.  U^{\prime}(W_{\tau})\right\vert \mathcal{H}%
_{s}\right)  \left(  \mathrm{d}P_{s}-\left(  \mathcal{D}_{s}P_{s}+\frac{1}%
{2}\bigtriangledown_{Y}^{2}P_{s}\sigma_{Z}^{2}(s)\right)  \mathrm{d}s\right)
\right\vert \mathcal{H}_{t}\right) \\
&  =K_{1}(t)\mathbb{E}\left(  \left.  \int_{t}^{\tau}\frac{1}{K_{1}%
(s)}\mathbb{E}\left(  \left.  U^{\prime}(W_{\tau})\right\vert \mathcal{H}%
_{s}\right)  \mathrm{d}P_{s}\right\vert \mathcal{H}_{t}\right)  -K_{1}%
(t)\mathbb{E}\left(  \left.  \int_{t}^{\tau}\frac{\mathbb{E}\left(  \left.
U^{\prime}(W_{\tau})\right\vert \mathcal{H}_{s}\right)  }{K_{1}(s)}\left(
\mathcal{D}_{s}P_{s}+\frac{1}{2}\bigtriangledown_{Y}^{2}P_{s}\sigma_{Z}%
^{2}(s)\right)  \mathrm{d}s\right\vert \mathcal{H}_{t}\right)  .
\end{align*}
Moreover, observe that
\begin{align*}
\int_{0}^{t}\frac{1}{K_{1}(s)}\mathbb{E}\left(  \left.  U^{\prime}(W_{\tau
})\right\vert \mathcal{H}_{s}\right)  \mathrm{d}P_{s}  &  =\frac
{\mathbb{E}\left(  \left.  U^{\prime}(W_{\tau})\right\vert \mathcal{H}%
_{t}\right)  P_{t}}{K_{1}(t)}-\frac{\mathbb{E}\left(  \left.  U^{\prime
}(W_{\tau})\right\vert \mathcal{H}_{0}\right)  P_{0}}{K_{1}(0)}\\
&  -\int_{0}^{t}P_{s}\mathrm{d}\left(  \frac{\mathbb{E}\left(  \left.
U^{\prime}(W_{\tau})\right\vert \mathcal{H}_{s}\right)  }{K_{1}(s)}\right)
-\left[  P,\frac{\mathbb{E}\left(  \left.  U^{\prime}(W_{\tau})\right\vert
\mathcal{H}_{\cdot}\right)  }{K_{1}(\cdot)}\right]  _{0}^{t}.
\end{align*}
Hence, taking (\ref{NC}) into account, we obtain
\begin{align*}
&  \mathbf{1}_{[0,\tau)}(t)\left(  \frac{\mathbb{E}\left(  \left.  U^{\prime
}(W_{\tau})\left(  V_{\tau}-P_{t}\right)  \right\vert \mathcal{H}_{t}\right)
}{K_{1}(t)}\right)  +\frac{\mathbb{E}\left(  \left.  U^{\prime}(W_{\tau
})\right\vert \mathcal{H}_{t\wedge\tau}\right)  P_{t\wedge\tau}}{K_{1}%
(t\wedge\tau)}-\frac{\mathbb{E}\left(  \left.  U^{\prime}(W_{\tau})\right\vert
\mathcal{H}_{0}\right)  P_{0}}{K_{1}(0)}\\
&  -\int_{0}^{t\wedge\tau}P_{s}\mathrm{d}\left(  \frac{\mathbb{E}\left(
\left.  U^{\prime}(W_{\tau})\right\vert \mathcal{H}_{s}\right)  }{K_{1}%
(s)}\right)  -\left[  P,\frac{\mathbb{E}\left(  \left.  U^{\prime}(W_{\tau
})\right\vert \mathcal{H}_{\cdot}\right)  }{K_{1}(\cdot)}\right]
_{0}^{t\wedge\tau}-\int_{0}^{t\wedge\tau}\frac{\mathbb{E}\left(  \left.
U^{\prime}(W_{\tau})\right\vert \mathcal{H}_{s}\right)  }{K_{1}(s)}\left(
\mathcal{D}_{s}P_{s}+\frac{1}{2}\bigtriangledown_{Y}^{2}P_{s}\sigma_{Z}%
^{2}(s)\right)  \mathrm{d}s\\
&  +\mathbb{E}\left(  \left.  \int_{0}^{\tau}\frac{\mathbb{E}\left(  \left.
U^{\prime}(W_{\tau})\right\vert \mathcal{H}_{s}\right)  }{K_{1}(s)}\left(
\mathcal{D}_{s}P_{s}+\frac{1}{2}\bigtriangledown_{Y}^{2}P_{s}\sigma_{Z}%
^{2}(s)\right)  \mathrm{d}s\right\vert \mathcal{H}_{t}\right)  -\mathbb{E}%
\left(  \left.  \int_{0}^{\tau}\frac{1}{K_{1}(s)}\mathbb{E}\left(  \left.
U^{\prime}(W_{\tau})\right\vert \mathcal{H}_{s}\right)  \mathrm{d}%
P_{s}\right\vert \mathcal{H}_{t}\right) \\
&  =0.
\end{align*}
Then by the uniqueness of the canonical decomposition in the previous equation
(notice that the jump of $\mathbf{1}_{[0,\tau)}(t)$ is killed in the case that
$\tau$ is predictive), we have
\begin{align*}
0  &  =\mathbb{E}\left(  \left.  U^{\prime}(W_{\tau})\left(  V_{\tau}%
-P_{t}\right)  \right\vert \mathcal{H}_{t}\right)  \frac{\mathrm{d}%
}{\mathrm{d}t}\left(  \frac{1}{K_{1}(t)}\right)  -\frac{\mathbb{E}\left(
\left.  U^{\prime}(W_{\tau})\right\vert \mathcal{H}_{t}\right)  }{K_{1}%
(t)}\left(  \mathcal{D}_{t}P_{t}+\frac{1}{2}\bigtriangledown_{Y}^{2}%
P_{t}\sigma_{Z}^{2}(t)\right) \\
&  -\frac{1}{K_{1}(t)}\frac{\mathrm{d}}{\mathrm{d}t}\left[  P,\mathbb{E}%
\left(  \left.  U^{\prime}(W_{\tau})\right\vert \mathcal{H}_{\cdot}\right)
\right]  _{0}^{t}+\frac{\mathrm{d}}{\mathrm{d}t}\left[  \mathbb{E}\left(
\left.  U^{\prime}(W_{\tau})\left(  V_{\tau}-P_{\cdot}\right)  \right\vert
\mathcal{H}_{\cdot}\right)  ,\frac{1}{K_{1}(\cdot)}\right]  _{0}^{t}.
\end{align*}
Finally the last term vanishes by Proposition \ref{factor}.
\end{proof}

Moreover, we have the following specific conditions in the risk-neutral and
risk-averse (exponential) cases.

\begin{proposition}
In the risk-neutral case, under the assumptions of Proposition \ref{neccond},
if $(P,X)$ is an equilibrium, then%
\[
\mathcal{D}_{t}P_{t}+\frac{1}{2}\bigtriangledown_{Y}^{2}P_{t}\sigma_{Z}%
^{2}(t)=0
\]
holds. Also, if $V_{t}\neq P_{t},$ a.s. $\mathbb{P\otimes}Leb,$ we have that
\begin{equation}
\mathcal{L}P_{t}=0. \label{lpt}%
\end{equation}

\end{proposition}

\begin{proof}
As a consequence of Proposition \ref{neccond} we have that, in the risk
neutral case, for the functionals above,
\[
0=\left(  V_{t}-P_{t}\right)  \frac{\mathrm{d}}{\mathrm{d}t}\left(  \frac
{1}{K_{1}(t)}\right)  -\frac{1}{K_{1}(t)}\left(  \mathcal{D}_{t}P_{t}+\frac
{1}{2}\bigtriangledown_{Y}^{2}P_{t}\frac{\mathrm{d}[Z]_{t}}{\mathrm{d}%
t}\right)
\]
By the competitiveness of prices $\mathbb{E}(V_{t}|\mathcal{F}_{t})$ $=P_{t}$,
so by taking conditional expectations w.r.t $\mathcal{F}_{t}$ we obtain that
\[
\mathcal{D}_{t}P_{t}+\frac{1}{2}\bigtriangledown_{Y}^{2}P_{t}\frac
{\mathrm{d}[Z]_{t}}{\mathrm{d}t}=0
\]
and if $V_{t}\neq P_{t},$ a.s. $\mathbb{P\otimes}Leb,$
\begin{equation}
\frac{\mathrm{d}}{\mathrm{d}t}\left(  \frac{1}{K_{1}(t)}\right)  =0.
\label{rnc}%
\end{equation}
Now by Proposition \ref{Bracket} we obtain (\ref{lpt}).
\end{proof}

Consider the risk-averse case when $U(x)=\gamma e^{\gamma x}$ with $\gamma<0$.
If the noise traders total demand $Z$ is Gaussian we can apply the following
relationship between vertical and Fréchet or Malliavin derivatives (see
Theorem 6.1 in \citeA{confou13}):
\[
\mathbb{E}\left(  \left.  D_{t}^{Z}U(W_{\tau})\right\vert \mathcal{H}%
_{t}\right)  =\nabla_{Z}\mathbb{E}\left(  \left.  U(W_{\tau})\right\vert
\mathcal{H}_{t}\right)  .
\]
Then by (\ref{NC})
\[
\mathbb{E}\left(  \left.  U^{^{\prime}}(W_{\tau})\left(  V_{\tau}%
-P_{t}\right)  +D_{t}^{Z}U(W_{\tau})\right\vert \mathcal{H}_{t}\right)  =0,
\]
we have that
\[
\nabla_{Z}\mathbb{E}\left(  \left.  U(W_{\tau})\right\vert \mathcal{H}%
_{t}\right)  =-\mathbb{E}\left(  \left.  U^{^{\prime}}(W_{\tau})\left(
V_{\tau}-P_{t}\right)  \right\vert \mathcal{H}_{t}\right)  .
\]
Since $U^{\prime}(x)=\gamma U(x)$, we obtain
\[
\nabla_{Z}\mathbb{E}\left(  \left.  U^{\prime}(W_{\tau})\right\vert
\mathcal{H}_{t}\right)  =-\gamma\mathbb{E}\left(  \left.  U^{^{\prime}%
}(W_{\tau})\left(  V_{\tau}-P_{t}\right)  \right\vert \mathcal{H}_{t}\right)
\]%
\[
\frac{\mathrm{d}\left[  P,\mathbb{E}\left(  \left.  U^{\prime}(W_{\tau
})\right\vert \mathcal{H}_{\cdot}\right)  \right]  }{\mathrm{d}t}=\nabla
_{Y}P_{t}\nabla_{Z}\mathbb{E}\left(  \left.  U^{^{\prime}}(W_{\tau
})\right\vert \mathcal{H}_{t}\right)  \sigma_{Z}^{2}(t)=K(t,t)\mathbb{E}%
\left(  \left.  U^{^{\prime}}(W_{\tau})\left(  V_{\tau}-P_{t}\right)
\right\vert \mathcal{H}_{t}\right)  \sigma_{Z}^{2}(t).
\]
Then (\ref{main}) becomes
\[
\mathbb{E}\left(  \left.  U^{\prime}(W_{\tau})\left(  V_{\tau}-P_{t}\right)
\right\vert \mathcal{H}_{t}\right)  \left(  \frac{\mathrm{d}}{\mathrm{d}%
t}\left(  \frac{1}{K_{1}}\right)  +\gamma K_{2}(t)\sigma_{Z}^{2}(t)\right)
-\frac{\mathbb{E}\left(  \left.  U^{\prime}(W_{\tau})\right\vert
\mathcal{H}_{t}\right)  }{K_{1}(t)}\left(  \mathcal{D}_{t}P_{t}+\frac{1}%
{2}\bigtriangledown^{2}P_{t}\sigma_{Z}^{2}(t)\right)  =0.
\]
Furthermore, if $V_{t}\equiv V$ we have that%
\[
\left(  V-P_{t}\right)  \left(  \frac{\mathrm{d}}{\mathrm{d}t}\left(  \frac
{1}{K_{1}}\right)  +\gamma K_{2}(t)\sigma_{Z}^{2}(t)\right)  -\frac{1}%
{K_{1}(t)}\left(  \mathcal{D}_{t}P_{t}+\frac{1}{2}\bigtriangledown_{Y}%
^{2}P_{t}\sigma_{Z}^{2}(t)\right)  =0.
\]
Taking the conditional expectations w.r.t $\mathcal{F}_{t}$, by the
competitiveness of prices $\mathbb{E}\left(  V|\mathcal{F}_{t}\right)
=P$\bigskip$_{t}$, we obtain that
\[
\mathcal{D}_{t}P_{t}+\frac{1}{2}\bigtriangledown_{Y}^{2}P_{t}\sigma_{Z}%
^{2}(t)=0.
\]
Provided that $V\neq P_{t},$ a.s. $\mathbb{P\otimes}Leb$, we have that
\begin{equation}
\frac{\mathrm{d}}{\mathrm{d}t}\left(  \frac{1}{K_{1}(t)}\right)  +\gamma
K_{2}(t)\sigma_{Z}^{2}(t)=0. \label{NRNcon}%
\end{equation}
Then we have the following proposition.

\begin{proposition}
Consider the risk-averse case with utility function is $U(x)=\gamma e^{\gamma
x},$ $\gamma<0$. Let $V_{t}\equiv V$ and assume that (\ref{CV})\ and
(\ref{CO}) hold . Also assume that $Z$ is Gaussian. If $(P,X)$ is an
equilibrium, we have \eqref{dP}:
\[
\mathcal{D}_{t}P_{t}+\frac{1}{2}\bigtriangledown_{Y}^{2}P_{t}\sigma_{Z}%
^{2}(t)=0
\]
and, if $V\neq P_{t},$ a.s. $\mathbb{P\otimes}Leb$, we have that
\[
\mathcal{L}P_{t}=\gamma\sigma_{Z}^{2}\left(  \nabla_{Y}P_{t}\right)  ^{2}.
\]

\end{proposition}

\begin{proof}
By (\ref{NRNcon})%
\[
\frac{1}{K_{1}^{2}(t)}\frac{\mathrm{d}}{\mathrm{d}t}K_{1}(t)=\gamma
K_{2}(t)\sigma_{Z}^{2}(t),
\]
now by Proposition \ref{Bracket} and Proposition \ref{KV}
\begin{align*}
\mathcal{L}P_{t}  &  =\gamma K_{1}^{2}(t)K_{2}^{2}(t)\sigma_{Z}^{2}(t)
=\gamma\left(  \nabla_{Y}P_{t}\right)  ^{2}\sigma_{Z}^{2}(t).
\end{align*}

\end{proof}

\begin{remark}
Finally, according to \citeA{codifa19} and for the above functionals if the
horizon $\tau$ is random and independent of the rest of processes involved, in
an equilibrium situation we have
\[
\frac{\mathrm{d}}{\mathrm{d}t}\frac{\mathbb{P(}\tau>t)}{K_{1}(t)}=0,
\]
then
\[
\frac{\mathrm{d}}{\mathrm{d}t}K_{1}(t)=K_{1}(t)\frac{\mathrm{d}}{\mathrm{d}%
t}\mathbb{P(}\tau>t),
\]
and by Proposition \ref{Bracket}
\begin{align*}
\mathcal{L}P_{t}  &  =K_{2}(t)\frac{\mathrm{d}K_{1}(t)}{\mathrm{d}t}%
=K_{1}(t)K_{2}(t)\frac{\mathrm{d}}{\mathrm{d}t}\mathbb{P(}\tau>t)\\
&  =\bigtriangledown_{Y}P_{t}\frac{\mathrm{d}}{\mathrm{d}t}\mathbb{P(}\tau>t).
\end{align*}

\end{remark}

\section{Examples of equilibrium pricing rules}

Consider the following class of functionals%
\[
P_{t}=H(t,\xi_{t}),\text{ }t\geq0,\text{ \ \ \ }\xi_{t}:=\int_{0}^{t}%
\lambda(s,P_{s})\mathrm{d}Y_{s},
\]
where $\lambda\in C^{1,2}$ is a strictly positive function and $H\in C^{1,3}$
with $H(t,\cdot)$ strictly increasing for every $t\geq0.$

Then, by using the It\^o's formula and omitting the arguments in the
functions, we have
\[
\mathrm{d}P_{t}=\partial_{2}H\lambda\mathrm{d}Y_{t}+\left(  \partial
_{1}H+\frac{1}{2}\partial_{22}H\lambda^{2}\sigma_{Z}^{2}\right)  \mathrm{d}t.
\]
Furthermore, we have that
\begin{align*}
\mathcal{D}_{t}P_{t}  &  =\partial_{1}H+\partial_{2}H\mathcal{D}_{t}\xi
_{t}=\partial_{1}H-\frac{1}{2}\partial_{2}H\bigtriangledown_{Y}^{2}\xi
_{t}\sigma_{Z}^{2}\\
&  =\partial_{1}H-\frac{1}{2}\partial_{2}H\bigtriangledown_{Y}\lambda
\sigma_{Z}^{2}\\
&  =\partial_{1}H-\frac{1}{2}\partial_{2}H\partial_{2}\lambda\bigtriangledown
_{Y}P_{t}\sigma_{Z}^{2}%
\end{align*}
and
\begin{align*}
\bigtriangledown_{Y}^{2}P_{t}  &  =\partial_{2}\lambda\bigtriangledown
_{Y}P_{t}\partial_{2}H+\lambda\partial_{22}H\bigtriangledown_{Y}\xi_{t}\\
&  =\partial_{2}\lambda\bigtriangledown_{Y}P_{t}\partial_{2}H+\lambda
^{2}\partial_{22}H.
\end{align*}
Consequently,
\[
\mathcal{D}_{t}P_{t}+\frac{1}{2}\bigtriangledown_{Y}^{2}P_{t}\sigma_{Z}%
^{2}=\partial_{1}H+\frac{1}{2}\partial_{22}H\lambda^{2}\sigma_{Z}^{2}.
\]
Then, under the condition
\[
\mathcal{D}_{t}P_{t}+\frac{1}{2}\bigtriangledown_{Y}^{2}P_{t}\sigma_{Z}%
^{2}=0,
\]
we have that%
\[
\partial_{1}H+\frac{1}{2}\partial_{22}H\lambda^{2}\sigma_{Z}^{2}=0,
\]
and by Proposition \ref{factor},
\[
K(s,t)=\frac{\lambda(s,P_{s})}{\eta_{s}}\partial_{2}H(t,\xi_{t})\eta_{t},
\]
where
\[
\eta_{t}:=\mathcal{E}\left(  \int_{0}^{t}\partial_{2}H\partial_{2}%
\lambda\mathrm{d}Y_{s}\right)  .
\]
Therefore $K(s,t)=K_{1}(s)K_{2}(t)$, with
\[
K_{1}(s)=\frac{\lambda(s,P_{s})}{\eta_{s}},K_{2}(t)=\partial_{2}H(t,\xi
_{t})\eta_{t}.
\]
By using the It\^o formula we obtain that
\begin{align*}
\mathrm{d}\left(  \frac{1}{K_{1}(s)}\right)   &  =\eta_{s}\mathrm{d}\left(
\frac{1}{\lambda}\right)  +\frac{1}{\lambda}\mathrm{d}\eta_{s}+\mathrm{d}%
\left[  \eta,\frac{1}{\lambda}\right]  _{s}\\
&  =-\eta_{s}\frac{\partial_{1}\lambda}{\lambda^{2}}\mathrm{d}s-\eta_{s}%
\frac{\partial_{1}\lambda}{\lambda^{2}}\partial_{2}H\lambda\mathrm{d}%
Y_{s}-\frac{1}{2}\eta_{s}\frac{\lambda^{2}\partial_{22}\lambda-2\left(
\partial_{1}\lambda\right)  ^{2}\lambda}{\lambda^{4}}\left(  \partial
_{2}H\right)  ^{2}\lambda^{2}\sigma_{Z}^{2}\mathrm{d}s\\
&  +\frac{1}{\lambda}\eta_{s}\partial_{1}\lambda\partial_{2}H\mathrm{d}%
Y_{s}-\frac{\left(  \partial_{1}\lambda\partial_{2}H\right)  ^{2}}{\lambda
}\sigma_{Z}^{2}\eta_{s}\mathrm{d}s\\
&  =-\eta_{s}\frac{\partial_{1}\lambda}{\lambda^{2}}\mathrm{d}s-\frac{1}%
{2}\partial_{22}\lambda\left(  \partial_{2}H\right)  ^{2}\sigma_{Z}^{2}%
\eta_{s}\mathrm{d}s\\
&  =-\eta_{s}\left(  \frac{1}{2}\partial_{22}\lambda\left(  \partial
_{2}H\right)  ^{2}\sigma_{Z}^{2}+\frac{\partial_{1}\lambda}{\lambda^{2}%
}\right)  \mathrm{d}s,
\end{align*}
Then, we have that
\begin{align*}
\mathcal{L}P_{t}  &  =K_{2}(t)\frac{\mathrm{d}}{\mathrm{d}t}K_{1}%
(t)=K_{2}(t)K_{1}^{2}(t)\eta_{t}\left(  \frac{1}{2}\partial_{22}\lambda\left(
\partial_{2}H\right)  ^{2}\sigma_{Z}^{2}+\frac{\partial_{1}\lambda}%
{\lambda^{2}}\right) \\
&  =\partial_{2}H\lambda^{2}\left(  \frac{1}{2}\partial_{22}\lambda\left(
\partial_{2}H\right)  ^{2}\sigma_{Z}^{2}+\frac{\partial_{1}\lambda}%
{\lambda^{2}}\right) \\
&  =\partial_{2}H\left(  \partial_{1}\lambda+\frac{1}{2}\sigma_{Z}^{2}\left(
\lambda\partial_{2}H\right)  ^{2}\partial_{22}\lambda\right)
\end{align*}
Hence, we will have an equilibrium price rule, in the risk-neutral case, if
\begin{align*}
\partial_{1}H+\frac{1}{2}\partial_{22}H\lambda^{2}\sigma_{Z}^{2}  &  =0,\\
\partial_{1}\lambda+\frac{1}{2}\sigma_{Z}^{2}\left(  \lambda\partial
_{2}H\right)  ^{2}\partial_{22}\lambda &  =0.
\end{align*}
and in the non risk-neutral case, for the exponential risk aversion, if
\[
\partial_{1}H+\frac{1}{2}\partial_{22}H\lambda^{2}\sigma_{Z}^{2}=0,
\]
and
\begin{align*}
\mathcal{L}P_{t}  &  =\partial_{2}H\left(  \partial_{1}\lambda+\frac{1}%
{2}\sigma_{Z}^{2}\left(  \lambda\partial_{2}H\right)  ^{2}\partial_{22}%
\lambda\right)  =\gamma\left(  \nabla_{Y}P_{t}\right)  ^{2}\sigma_{Z}^{2}\\
&  =\gamma K_{1}^{2}(t)K_{2}^{2}(t)\sigma_{Z}^{2}=\gamma\left(  \partial
_{2}H\lambda\right)  ^{2}\sigma_{Z}^{2}%
\end{align*}
that is
\begin{equation}
\frac{\partial_{1}\lambda}{\lambda^{2}}+\frac{1}{2}\sigma_{Z}^{2}\left(
\partial_{2}H\right)  ^{2}\partial_{22}\lambda=\gamma\partial_{2}H\sigma
_{Z}^{2}. \label{2condnon}%
\end{equation}

We can identify some particular cases.

\noindent\textbf{For the risk-neutral case}

\begin{itemize}
\item $\ \lambda(t,x)=\lambda>0,$ and $H(t,x)$ harmonic with $H(,x)$ strictly
increasing$.$ Notice that in this case it is sufficient to require that
$H(t,x)$ is $\mathcal{C}^{1,2}$.

\item If we take $H(t,x)=x$ and $\lambda(t,x)=\lambda>0$, we have
\[
P_{t}=P_{0}+\lambda Y_{t},
\]
that corresponds to the Bachelier model for $Z$ Gaussian.

\item If $H(t,x)=x$ and $\lambda(t,x)=\lambda x,$ we have
\[
P_{t}=P_{0}e^{\lambda Y_{t}-\frac{1}{2}\lambda^{2}t}%
\]
that is the Black-Scholes model.
\end{itemize}

\noindent\textbf{For the non risk-neutral model}

\begin{itemize}
\item Note that $H(t,x)$ harmonic and $\lambda$ constant cannot be an
equilibrium. Therefore equilibrium prices cannot be a function of the spot
aggregate demand.

\item If we take $H(t,x)=x$ and $\lambda(t,x)=Cx(1-x)$, $\ $with $C>0,$ we
have that (\ref{2condnon}) becomes
\[
\frac{1}{2}\partial_{xx}\lambda=\gamma
\]
that is $\gamma=-C.$ This model will give prices in $\left(  0,1\right)  $ and
if $Y$ is a Brownian motion $B$ we have that%
\[
\mathrm{d}P_{t}=CP_{t}(1-P_{t})\mathrm{d}B_{t}%
\]
and this is the well-known Kimura model in population genetics, see \citeA{kimura64}.
\end{itemize}

\section{Examples of equilibrium models}

It is apparent that depending on the behaviour of the fundamental value and
the aggregate demand of the noise traders we can have an equilibrium with one
or another equilibrium pricing rule. If the aggregate demand $Z$ of the noise
traders is a Brownian motion with variance \ $\sigma_{Z}^{2}$, $Y$ $=X+Z$ will
be also an $\mathbb{F}$-Brownian motion with variance $\sigma_{Z}^{2}$,
because of Theorem \ref{T1} and the L\'evy Theorem. Consequently, we will have
an equilibrium if the pricing rule $P_{t}(Y_{\cdot t})$ is such that
$P_{T}(Y_{\cdot T})=V_{T}.$ Note also that the strategy $X$ will be just
obtained as the canonical decomposition of the $\mathbb{F}$-Brownian motion
$Y$ under the filtration $\mathbb{H}$.

Consider the case where $Z$ is a Brownian motion with variance $\sigma^{2}$
and $V_{t}\equiv V$. In such a situation we have a necessary and sufficient
condition for and equilibrium for both, the risk-neutral case and the
risk-adverse case under the exponential utility. Also in both cases the
equilibrium pricing rules give prices that are continuous diffusions:%
\[
\mathrm{d}P_{t}=\lambda(t,P_{t})\mathrm{d}Y_{t}%
\]
where $\mathrm{d}Y_{t}=\sigma\mathrm{d}W_{t}$ and $W$ is a standard Brownian
motion. In the risk-neutral case $\lambda(t,x)$ satisfies%
\[
\partial_{t}\lambda+\frac{1}{2}\lambda^{2}\sigma^{2}\partial_{xx}\lambda=0,
\]
and in the risk-adverse case
\[
\partial_{t}\lambda+\frac{1}{2}\lambda^{2}\sigma^{2}\partial_{xx}%
\lambda=\gamma\lambda^{2}\sigma^{2}.
\]
In any case the additional necessary and sufficient condition to have an
equilibrium model is to find a strategy such that $P_{T}=V$ and at the same
time $Y$ is certainly a Brownian motion with variance $\sigma^{2}.$ We have to
find $\alpha_{t}(V),$with $\alpha_{t}(x)$ $\mathcal{F}_{t}$-measurable,\ such
that the equation
\[
\mathrm{d}Y_{t}=\alpha_{t}(V)\mathrm{d}t+\mathrm{d}Z_{t},\text{ }0\leq t\leq
T,
\]
with $V=P_{T}$ and independent of $Z,$ has a strong solution. In order to do
so, we can look for certain $\alpha(t,x,Y_{\cdot t}),$ where $x$ is a fixed
value of $P_{T}$ and try to find a strong solution of
\[
\mathrm{d}Y_{t}=\alpha_{t}(x)\mathrm{d}t+\mathrm{d}Z_{t},\text{ }0\leq t\leq
T,
\]
later we can insert $V$ instead of $y$, but we need $Y$ to be a Brownian
motion with variance $\sigma^{2}$. Sufficient conditions to have a strong
solution are given, e.g., in Theorem 4.6, \citeA{lipshir01}. Then $\alpha
_{t}(x)$ has to be the drift in the canonical decomposition of $Y$ when
$Y_{T}$ and $Z_{\cdot t}$ are known at time $t.$ The following propositions
are useful to find $\alpha,$ here we assume that $\mathcal{F}_{t}=\bar{\sigma
}\left(  W_{\cdot t},t\geq0\right)  ,$ $\bar{\sigma}$ denotes the $\sigma
$-field corresponding to the usual augmentation of the natural filtration.

\begin{proposition}
Assume that for any bounded and measurable function $f$ there exists a
$\mathcal{B}[0,T])\otimes\mathcal{F}_{T}$-measurable process $\xi,$
independent of $f$, such that
\[
f(P_{T})=\mathbb{E(}f(P_{T})\mathbb{)+}\int_{0}^{T}\mathbb{E}\left(
f(P_{T})\xi_{t}|\mathcal{F}_{t}\right)  \mathrm{d}W_{t},
\]
with $\int_{0}^{T}\left\vert \xi_{t}\right\vert \mathrm{d}t<\infty$. Then
$W_{\cdot}-\int_{0}^{\cdot}\alpha_{t}(P_{T})\mathrm{d}t$ is an $\left(
\mathcal{F}_{t}\vee\sigma(P_{T})\right)  $-Brownian motion with
\[
\alpha_{t}(P_{T})=\mathbb{E}\left(  \xi_{t}|\mathcal{F}_{t}\vee\sigma
(P_{T})\right)
\]

\end{proposition}

\begin{proof}
Let $f$ \ be a measurable and bounded function and $A\in\mathcal{F}_{s},$ with
$s\leq t.$ Then
\begin{align*}
\mathbb{E}\left(  \left(  W_{t}-W_{s}\right)  \mathbf{1}_{A}f(P_{T})\right)
&  =\mathbb{E}\left(  \mathbf{1}_{A}\int_{s}^{t}\mathbb{E}\left(  f(P_{T}%
)\xi_{u}|\mathcal{F}_{u}\right)  \mathrm{d}u\right) \\
&  =\mathbb{E}\left(  \mathbf{1}_{A}f(P_{T})\int_{s}^{t}\mathbb{E}\left(
\xi_{u}|\mathcal{F}_{u}\vee\sigma(P_{T})\right)  \mathrm{d}u\right)  .
\end{align*}

\end{proof}

\begin{proposition}
Suppose that
\[
\mathrm{d}\mathbb{P}_{P_{T}|\mathcal{F}_{t}}(x|\mathcal{F}_{t})=L_{T}%
(x;W_{\cdot t})\mu\left(  \mathrm{d}x\right)
\]
is a regular version of the conditional probability of $P_{T}$ given
$\mathcal{F}_{t}$, $\mu$ being a reference measure and such that
\begin{align*}
&  \left.  i)\text{ }L_{T}(x;W_{\cdot t})>0\text{ for all }\left(
x,\omega\right)  \text{ }\mu\otimes\mathbb{P}\text{-a.s.},\right. \\
&  \left.  ii)\text{ }\nabla_{W}\int_{\mathbb{R}}f(x)L_{T}(x;W_{\cdot t}%
)\mu\left(  \mathrm{d}x\right)  =\int_{\mathbb{R}}f(x)\nabla_{W}%
L_{T}(x;W_{\cdot t})\mu\left(  \mathrm{d}x\right)  \right.  .
\end{align*}
Then $W_{\cdot}-\int_{0}^{\cdot}\alpha_{t}(P_{T})\mathrm{d}t$ is an $\left(
\mathcal{F}_{t}\vee\sigma(P_{T})\right)  $-Brownian motion with $\alpha
_{t}(x)=\nabla_{W}\log L_{T}(x;W_{\cdot t})$, provided that $\log
L_{T}(x;W_{\cdot t})\in\mathbb{C}^{1}.$
\end{proposition}

\begin{proof}
Let $f$ \ be a measurable and bounded function
\begin{align*}
\nabla_{W}\mathbb{E}\left(  f(P_{T})|\mathcal{F}_{t}\right)   &  =\nabla
_{W}\int_{\mathbb{R}}f(x)L_{T}(x;W_{\cdot t})\mu\left(  \mathrm{d}x\right) \\
&  =\int_{\mathbb{R}}f(x)\nabla_{W}L_{T}(x;W_{\cdot t})\mu\left(
\mathrm{d}x\right) \\
&  =\int_{\mathbb{R}}f(x)\nabla_{W}\log L_{T}(x;W_{\cdot t})L_{T}(x;W_{\cdot
t})\mu\left(  \mathrm{d}x\right) \\
&  =\mathbb{E}\left(  f(P_{T})\nabla_{W}\log L_{T}(P_{T};W_{\cdot
t})|\mathcal{F}_{t}\right)  .
\end{align*}

\end{proof}

\begin{example}
Assume that $P_{t}=P_{0}+\sigma W_{t}.$ Then $P_{T}|\mathcal{F}_{t}\sim
$N$\left(  P_{0}+\sigma W_{t},\sigma^{2}(T-t)\right)  $, that is
\[
\mathrm{d}\mathbb{P}_{P_{T}|\mathcal{F}_{t}}(x|\mathcal{F}_{t})=\frac{1}%
{\sqrt{2\pi\sigma^{2}(T-t)}}\exp\left\{  -\frac{1}{2\sigma^{2}(T-t)}\left(
x-P_{0}-\sigma W_{t}\right)  ^{2}\right\}  \mathrm{d}x,
\]
then
\[
\alpha_{t}(x)=\frac{\sigma\left(  x-P_{0}-\sigma W_{t}\right)  }{\sigma
^{2}(T-t)},
\]
that is
\[
\alpha_{t}(P_{T})=\frac{W_{T}-W_{t}}{T-t}.
\]

\end{example}

\begin{example}
Assume that $P_{t}=P_{0}+\int_{0}^{t}G(u,P_{u})\mathrm{d}W_{u}.$
$G\in\mathcal{C}^{1,2},\mathbb{E}\left(  \int_{0}^{T}G^{2}(t,P_{t}%
)\mathrm{d}t\right)  <\infty.$ Let $p_{s,t}(x,y)$ the transition density
corresponding to the Markov process $P$. Then according to the previous
proposition
\begin{align*}
\alpha_{t}(y)  &  =\nabla_{W}\log p_{t,T}(P_{t},y)\\
&  =\partial_{1}\log p_{t,T}(P_{t},y)\nabla_{W}P_{t}\\
&  =\partial_{1}\log p_{t,T}(P_{t},y)G(t,P_{t}).
\end{align*}
For instance, we can consider the simple case where $P_{t}=P_{0}+\int_{0}%
^{t}\sigma P_{u}\mathrm{d}W_{u}$, then, for $s\leq t$, $P_{t}=P_{s}%
\exp\left\{  \sigma\left(  W_{t}-W_{s}\right)  -\frac{1}{2}\sigma
^{2}(t-s)\right\}  ,$ and
\[
p_{s,t}(x;y)=\frac{1}{\sqrt{2\pi\sigma^{2}(T-t)}}\exp\left\{  -\frac
{1}{2\sigma^{2}(T-t)}\left(  \log y-\log x-\frac{1}{2}\sigma^{2}(t-s)\right)
^{2}\right\}  \frac{1}{y},
\]
consequently%
\[
\alpha_{t}(y)=\partial_{1}\log p_{t,T}(P_{t},y)G(t,P_{t})=\sigma\frac{\log
y-\log P_{t}-\frac{1}{2}\sigma^{2}(t-s)}{\sigma^{2}(T-t)},
\]
that is
\[
\alpha_{t}(P_{T})=\frac{W_{T}-W_{t}}{T-t}.
\]
Notice that in this case we obtain the same result as in the previous example.
This is not surprising since in both cases to know $P_{T}$ is the same as to
know $W_{T}$ since $P_{T}$ is an increasing function of $W_{T}.$ Obviously
this will not be the case for a general diffusion. The simplest case where
this\ does not happen is the equilibrium price model
\[
P_{t}=P_{0}+\int_{0}^{t}\lambda(s)\mathrm{d}W_{s},
\]
with
\[
\partial_{t}\lambda=\gamma\lambda^{2}\sigma^{2}.
\]
arising in the non-risk-neutral model. Now $P_{T}|P_{t}\sim N\left(
P_{0},\int_{0}^{t}\lambda^{2}(s)\mathrm{d}s\right)  $, then%
\begin{align*}
\alpha_{t}(x)  &  =\nabla_{W}\log L_{T}(x;W_{\cdot t})\\
&  =\frac{\lambda(t)\left(  x-P_{t}\right)  }{\int_{t}^{T}\lambda
^{2}(s)\mathrm{d}s}.
\end{align*}
If we consider the Kimura model, with risk-aversion parameter $\gamma<0$,
\[
P_{t}=P_{0}-\gamma\int_{0}^{t}P_{t}(1-P_{t})\mathrm{d}W_{t},\text{ }%
\]
then the transition density is given by \cite{kimura64}
\[
p_{t,T}(P_{t},x)=\frac{1}{\sqrt{2\pi\gamma^{2}(T-t)}}\frac{\sqrt{P_{t}%
(1-P_{t})}}{\left(  \sqrt{x(1-x)}\right)  ^{3}}\exp\left\{  -\frac{\gamma^{2}%
}{8}(T-t)-\frac{\left(  \log\frac{x(1-P_{t})}{(1-x)P_{t}}\right)  ^{2}%
}{2\gamma^{2}(T-t)}\right\}  ,
\]
and we have that
\begin{align*}
\alpha_{t}(y)  &  =\nabla_{W}\log L_{T}(x;W_{\cdot t})\\
&  =\frac{1}{2}(1-2P_{t})+\frac{\log\frac{x(1-P_{t})}{(1-x)P_{t}}}{\gamma
^{2}(T-t)}.
\end{align*}

\end{example}

\begin{example}
We can consider the case where the privilege information is the time, say
$\tau$, where a Brownian motion reaches for the first time a level $a.$ Then,
assume that
\[
P_{T}=h(T\wedge\tau)
\]
for a measurable and bounded function $h$. Now we have that
\[
P_{T}=P_{0}+\int_{0}^{T}\nabla_{W}\mathbb{E}\left(  h(T\wedge\tau
)|\mathcal{F}_{t}\right)  \mathrm{d}W_{t},
\]
then since
\[
f_{\tau}(u|\mathcal{F}_{t})=\frac{W_{t}-a}{\sqrt{2\pi(u-t)^{3}}}\exp\left\{
-\frac{\left(  W_{t}-a\right)  ^{2}}{2(u-t)}\right\}  \mathbf{1}_{\{\tau
>t\}},
\]
we obtain that
\[
\alpha_{t}(u)=\nabla_{W}\log f_{\tau}(u|\mathcal{F}_{t})=\left(  \frac
{1}{W_{t}-a}-\frac{W_{t}-a}{u-t}\right)  \mathbf{1}_{\{\tau>t\}}.
\]

\end{example}

\bigskip

\textbf{Acknowledgement}: This paper is devoted to the memory of our beloved
colleague Jos\'e Fajardo Barbach\'an, who passed away before we could complete
together this work.

\bibliographystyle{apacite}
\bibliography{referapa}

\end{document}